\documentclass{IEEEtran}
\usepackage{amsmath}
\usepackage{nccmath}
\usepackage[utf8]{inputenc} %unicode support
\usepackage{epsfig,scalefnt,multirow}
\usepackage{url}
\usepackage{ifthen}
\usepackage{mathtools}
\usepackage{cite}
\usepackage{graphicx}
\usepackage{amssymb}
\usepackage{tabularx}
\usepackage{amsmath}
\usepackage{epstopdf}
\usepackage{epsf}
\usepackage{algorithm}
\usepackage{algpseudocode}
\usepackage{algpascal}
\usepackage{cases}
\usepackage{caption}
\usepackage[labelformat=simple]{subcaption}

\usepackage{stfloats}
\usepackage{float}
\usepackage{xcolor}
\usepackage{tabularx}% http://ctan.org/pkg/tabularx

\usepackage{epsfig,amsmath,amssymb,epsf,amsthm,scalefnt,multirow}
\usepackage{xcolor}
\usepackage{float}
\usepackage{cite}
\usepackage{psfrag}
\usepackage{gensymb}
\usepackage{cleveref}
\usepackage{mathrsfs}
\usepackage{mathtools}
\usepackage{algpseudocode}
\usepackage{algorithm}
\usepackage{enumerate}

  \def\cC{{\mathcal{C}}} 
   
  \def\cK{{\mathcal{K}}} 
 \def\cN{{\mathcal{N}}}  
  \def\cS{{\mathcal{S}}} 
 \def\cV{{\mathcal{V}}}

\def\argmax{\mathop{\mathrm{argmax}}}

\def\Re{\mathop{\mathrm{Re}}}
\def\Im{\mathop{\mathrm{Im}}}

\def\bSigma{{\pmb{\Sigma}}}

\def\bOmega{{\pmb{\Omega}}}

 \def\bmu{{\pmb{\mu}}}
\def\b0{{\pmb{0}}} 

\def\ba{{\mathbf{a}}} \def\bb{{\mathbf{b}}} \def\bc{{\mathbf{c}}} 
 \def\bff{{\mathbf{f}}} \def\bg{{\mathbf{g}}} \def\bh{{\mathbf{h}}}
   
\def\bm{{\mathbf{m}}}   
  \def\bs{{\mathbf{s}}} \def\bt{{\mathbf{t}}}
\def\bu{{\mathbf{u}}}  \def\bw{{\mathbf{w}}} \def\bx{{\mathbf{x}}}
\def\by{{\mathbf{y}}} \def\bz{{\mathbf{z}}}  

\def\bA{{\mathbf{A}}}   
 \def\bF{{\mathbf{F}}}  
\def\bI{{\mathbf{I}}}

\DeclarePairedDelimiter\norm{\lVert}{\rVert}

\newtheorem{lemma}{Lemma}
\newtheorem*{lemma*}{Lemma}

\begin{document}
	%
	% paper title
	% Titles are generally capitalized except for words such as a, an, and, as,
	% at, but, by, for, in, nor, of, on, or, the, to and up, which are usually
	% not capitalized unless they are the first or last word of the title.
	% Linebreaks \\ can be used within to get better formatting as desired.
	% Do not put math or special symbols in the title.

	\title{Coverage Increase at THz Frequencies: \\ A Cooperative Rate-Splitting Approach}

	%
	%
	% author names and IEEE memberships
	% note positions of commas and nonbreaking spaces ( ~ ) LaTeX will not break
	% a structure at a ~ so this keeps an author's name from being broken across
	% two lines.
	% use \thanks{} to gain access to the first footnote area
	% a separate \thanks must be used for each paragraph as LaTeX2e's \thanks
	% was not built to handle multiple paragraphs
	%
	
	\author{Hyesang Cho,~\IEEEmembership{Student Member,~IEEE}, Beomsoo Ko,~\IEEEmembership{Student Member,~IEEE}, Bruno Clerckx,~\IEEEmembership{Fellow,~IEEE},	and~Junil Choi,~\IEEEmembership{Senior Member,~IEEE}% <-this % stops a space
		\thanks{H. Cho, B. Ko, and J. Choi are with the School of Electrical Engineering, Korea Advanced Institute of Science and Technology, Daejeon 34141, South Korea (e-mail: \{nanjohn96, kobs0318, junil\}@kaist.ac.kr).
			
		B. Clerckx is with the Department of Electrical and Electronic Engineering, Imperial College London, London SW7 2AZ, U.K (e-mail: b.clerckx@imperial.ac.uk).
		}% <-this % stops a space
		%\thanks{Junil Choi is the corresponding author.}
	}

	\maketitle
	
	% As a general rule, do not put math, special symbols or citations
	% in the abstract or keywords.

	\begin{abstract} \label{sec:abs} % abstract
		Numerous studies claim that terahertz (THz) communication will be an essential piece of sixth-generation wireless communication systems. 
		Its promising potential also comes with major challenges, in particular the reduced coverage due to harsh propagation loss, hardware constraints, and blockage vulnerability. 
		To increase the coverage of THz communication, we revisit cooperative communication. 
		We propose a new type of cooperative rate-splitting (CRS) called extraction-based CRS (eCRS).
		Furthermore, we explore two extreme cases of eCRS, namely, identical eCRS and distinct eCRS.	
		To enable the proposed eCRS framework, we design a novel THz cooperative channel model by considering unique characteristics of THz communication.
		Through mathematical derivations and convex optimization techniques considering the THz cooperative channel model, we derive local optimal solutions for the two cases of eCRS and a global optimal closed form solution for a specific scenario. 
		Finally, we propose a novel channel estimation technique that not only specifies the channel value, but also the time delay of the channel from each cooperating user equipment to fully utilize the THz cooperative channel.
		In simulation results, we verify the validity of the two cases of our proposed framework and channel estimation technique.
	\end{abstract}
	
	%%%%%%%%%%%%%%%%%%%%%%%%%%%%%%%%%%%%%%%%%%%%%%
	%%                                          %%
	%% The keywords begin here                  %%
	%%                                          %%
	%% Put each keyword in separate \kwd{}.     %%
	%%                                          %%
	%%%%%%%%%%%%%%%%%%%%%%%%%%%%%%%%%%%%%%%%%%%%%%
	
	
	% MSC classifications codes, if any
	%\begin{keyword}[class=AMS]
	%\kwd[Primary ]{}
	%\kwd{}
	%\kwd[; secondary ]{}
	%\end{keyword}

	%\end{fmbox}% uncomment this for twcolumn layout

	%%%%%%%%%%%%%%%%%%%%%%%%%%%%%%%%%%%%%%%%%%%%%%
	%%                                          %%
	%% The Main Body begins here                %%
	%%                                          %%
	%% Please refer to the instructions for     %%
	%% authors on:                              %%
	%% http://www.biomedcentral.com/info/authors%%
	%% and include the section headings         %%
	%% accordingly for your article type.       %%
	%%                                          %%
	%% See the Results and Discussion section   %%
	%% for details on how to create sub-sections%%
	%%                                          %%
	%% use \cite{...} to cite references        %%
	%%  \cite{koon} and                         %%
	%%  \cite{oreg,khar,zvai,xjon,schn,pond}    %%
	%%  \nocite{smith,marg,hunn,advi,koha,mouse}%%
	%%                                          %%
	%%%%%%%%%%%%%%%%%%%%%%%%%%%%%%%%%%%%%%%%%%%%%%
	
	%%%%%%%%%%%%%%%%%%%%%%%%% start of article main body
	% <put your article body there>
	\section{Introduction}\label{sec:intro}
	Along with the constant evolution of modern technology, wireless communication has steadily developed to satisfy increasing demands of high data traffic in fifth-generation wireless communication systems. To satisfy the demands, wireless communication systems have implemented novel technologies such as massive multiple-input multiple-output (MIMO) systems and millimeter-wave technology \cite{mmWave1, mmWave2, mmWave3}. While these technologies are sufficient for current requirements, innovative approaches are in need to meet the requirements of sixth-generation (6G) systems \cite{SAM}.
	
	One promising technology for 6G systems is terahertz (THz) communication, which exploits broad and vacant frequency bands ranging from $0.1$ THz to $10$ THz \cite{THz1, THz2, THz3}. While a primitive THz system has already succeeded in $100$ Gbps communication \cite{THzbad3}, THz communication is still in its infancy due to limitations in both poor channel conditions and low signal power. In THz communication, the propagation loss degrades the signal power in orders of magnitude larger than the conventional frequency bands, and molecular absorption reduces the signal power even further, where the degradation is exponentially proportional with respect to distance \cite{THzbad1, THzbad2}. Also, direct line-of-sight (LoS) links become vulnerable to blockage due to minimal scatterings and decreased multipaths.  Furthermore, limited hardware specifications are additional factors to surmount. Although various approaches to generate the THz frequency signals were studied, the signal generators suffer from low transmit power \cite{THzbad3, THzbad4, THzbad5}. The THz frequency signals were generated using photonic devices and photomixing in \cite{THzbad3}, but the transmit power was extremely low due to the limited gain of the uni-travelling-carrier photodiode. While the non-photonic device was used to generate signals of $300$ GHz carrier frequency in \cite{THzbad4}, it was inadequate for the high power regime since the power conversion had linearity only in the low transmit power regime.

	There have been several studies considering unique characteristics of THz communication. Considering that the far-field assumption does not apply to short-range wireless communication in THz frequencies, it was shown in \cite{THzwork1} that multiplexing in the direct LoS environment is possible, where reconfigurable array architectures were proposed to exploit this effect. Other works considered relaying systems, where the studies focused on increasing the performance from the detrimental channels due to the molecular absorption \cite{THzwork2, THzwork3}.
	
	In this paper, we revisit the concept of cooperative communication to overcome the harsh channel conditions and low signal power devices in THz communication. Cooperative communication is an idea that gained attention in the early 2000s, and the basic concept is to cooperate among user equipments (UEs) to gain better performance for the overall system, e.g., the strong UE can additionally use its power to increase the coverage of the system by assisting weak UEs \cite{Coop1, Coop2}. To be clear, we distinguish the cooperative communication system from the relay system. In the relay system, the relaying devices only transmit messages for other UEs, while in the cooperative communication system, the cooperating UEs not only transmit but also receive messages for themselves. While the concept was promising, one main weakness of cooperative communication was the selfish behavior of the UEs. In practice, no UE would want to use its resource to help other UEs. However, nowadays, people possess multiple devices, and the future is envisaged as an Internet-of-Things society. Thus, we can strongly anticipate the increase of devices \cite{Future1}. By utilizing multiple devices, cooperative communication can be successful for future wireless communication such as 6G.
	
	To effectively support cooperative communication, we need to adopt a proper multiple access technique, e.g., non-orthogonal multiple access (NOMA), spatial-division multiple access (SDMA), or rate-splitting multiple access (RSMA) \cite{NOMA1, SDMA1, RSMA6}. In NOMA, a successive interference cancellation (SIC) technique is adopted to fully decode the interference, while in SDMA, the interference is fully considered as noise in general. 
	RSMA is a superset of NOMA and SDMA and bridges those two schemes to further increase the performance.
	This is achieved by enabling RSMA to partially decode interference, thanks to the creation of a common stream and the presence of SIC, and partially treating the remaining interference as noise \cite{RSMA1}.
	To that end, RSMA divides the messages of the UEs into private and common parts, where the private parts are separately encoded into private streams, and the common parts are jointly encoded into a common stream. From the UE standpoint, the common stream is decoded first while the private streams are considered as noise. Then, by removing the decoded common stream, the desired private stream is decoded with reduced interference.
	
	It has been shown that RSMA includes both SDMA and NOMA as extreme cases and has superior performance in both perfect channel state information (CSI) and imperfect CSI cases \cite{RSMA6,RSMA1}. 
	Due to its superior characteristics, RSMA gained attention in many studies and was implemented in  various applications to increase performance \cite{RSMA6}. 
	For example, the concept of cooperative rate-splitting (CRS) was proposed in \cite{RSMA2, RSMA3 ,RSMA7}, where the cooperating UEs relay the common message to increase the minimum rate of the UEs. The sum rate and energy efficiency was improved in a non-orthogonal unicast multicast (NOUM) problem with RSMA \cite{RSMA4}.
	
	In this paper, we consider a THz multiple-user (MU) multiple-input single-output (MISO) downlink system, where an access point (AP) is not able to communicate with a UE, denoted as a destination UE (dUE), due to, for example, a blocked channel.
	The other UEs, denoted as medium UEs (mUEs), will then cooperate by sharing their resources to achieve reliable communication for the dUE while the mUEs also decode their own messages as well. 
	By introducing a two-phase cooperative communication system, the AP transmits messages to the mUEs in the first phase, and the mUEs transmit messages to the dUE in the second phase. 
	The contributions of this paper are summarized as follows:
	\begin{itemize}
		\item For the two-phase cooperative communication, we propose a new type of CRS, named extraction-based CRS (eCRS). 
		The proposed eCRS has two main differences compared to conventional CRS.
		The first difference appears when the mUEs relay the messages to the dUE in the second phase.
		In conventional CRS, the common message that includes not only the message for the dUE, but also the messages for the mUEs is directly relayed to the dUE as in \cite{RSMA2, RSMA3, RSMA7}.
		This may cause an inefficient use of resources since the mUEs transmit messages unrequired for the dUE.
		In contrast, in eCRS, each mUE extracts the message for the dUE from the common message, and then relays only the message for the dUE in the second phase. 
		Through this approach, the mUEs can use their power solely for the message of the dUE.
		The second difference appears in the two-phase transmission framework, which explains how the messages are split and combined.
		While conventional CRS is only based on a single common message \cite{RSMA2, RSMA3, RSMA7}, we establish a transmission framework for eCRS by employing multiple common messages, similar to the general RSMA \cite{RSMA1}.
		
		\item We also investigate two extreme cases of eCRS, namely, identical eCRS (IeCRS) and distinct eCRS (DeCRS), where both cases only require one SIC layer. 
		IeCRS, which is the same as NOUM in the first phase, transmits an identical common message to the mUEs, and then the mUEs decode-and-forward the identical message to the dUE.
		In contrast, DeCRS transmits distinct common messages to each mUE, and then each mUE decode-and-forwards the distinct message to the dUE.
%		Since every mUE has to decode the identical common message in IeCRS, the performance can be bounded by the mUE with the worst channel condition. However, in DeCRS, since each mUE decodes distinct common message, multiple common messages can be adopted to the heterogeneous channel conditions of the mUEs.   
		Optimization problems are formulated for both IeCRS and DeCRS. 
		Through mathematical derivations and convex optimization techniques, we show that IeCRS and DeCRS are both effective in different scenarios.
	
		\item We develop a THz cooperative channel model by considering unique characteristics of THz communication. 
		In conventional cooperative communication, it is assumed that the transmit signals from the multiple mUEs arrive at the dUE in a single tap.
		However, for cooperative communication in THz communication systems, the signals from the mUEs arrive in different taps due to the large bandwidth and high sampling rate.
		By judiciously taking the multipaths through the mUEs into account, the THz cooperative channel can obtain additional gain, which has not been considered in existing literature.
%		The THz cooperative channel can obtain additional gain by judiciously taking the multipaths through the mUEs into account, which has not been considered in existing literature.
		
		\item We propose a novel channel estimation technique to fully utilize the developed THz cooperative channel model. 
		While the channel gain may be obtainable with conventional channel estimation techniques for frequency selective channels, the THz cooperative channel model requires additional information of both the channel gain and time delay for each specific mUE.
		By utilizing the characteristics of sequences such as pseudo-noise or Zadoff-chu sequences, we derive an estimation technique to not only estimate the channel gain, but also identify the delays of all the mUEs. 
		In result, we verify that the channel estimation technique works sufficiently well even in the low transmit power regime and also show the performances of the two cases of our proposed framework with imperfect CSI.
	\end{itemize}

	The rest of paper is organized as follows. 
	Section \ref{sec:model} describes the transmission framework of eCRS and details of IeCRS and DeCRS.
	In Section \ref{sec:IDT}, we formulate and solve the minimum rate maximization problem for IeCRS through convex optimization techniques. 
	We also derive an effective closed form solution with minimal complexity and optimal performance for IeCRS. 
	In Section \ref{sec:DDT}, we formulate and solve the optimization problem for DeCRS. 
	Section \ref{sec: est} describes the specifics of the channel estimation technique adequate for the THz cooperative channel model. 
	Section \ref{sec: simul} shows the simulation results of the channel estimation technique and cooperative communication techniques with perfect and imperfect CSI. 
	Finally, Section \ref{sec:concl} concludes our paper.
	
	\textbf{Notation:} Lower and upper boldface letters represent column vectors and matrices. $\bA^{*}$  and $\bA^{\mathrm{H}}$ denote the conjugate, and conjugate transpose of the matrix $\bA$. ${\mathbb{C}}^{m \times n}$ and ${\mathbb{R}}^{m \times n}$ represent the set of all $m \times n$ complex and real matrices. $|{\cdot}|$ denotes the amplitude of the scalar, and $\norm{\cdot}$ represents the $\ell_2$-norm of the vector. $\mathbb{Z}$ denotes the set of integers.  $\mathcal{O}$ denotes the Big-O notation. The Kronecker product is denoted by $\otimes$.
	$\boldsymbol{0}_m$ is used for the $m\times1$ all zero vector, and $\bI_m$ denotes the $m \times m$ identity matrix. $\cC \cN(\bm,\bSigma)$ denotes the circularly symmetric complex Gaussian distribution with mean $\bm$ and variance $\bSigma$.
	
	% \begin{figure}[t] 
	%	\centering
	%	\includegraphics[width=1 \columnwidth]{SystemMod.eps}
	%	%   % where an .eps filename suffix will be assumed under latex,
	%	%   % and a .pdf suffix will be assumed for pdflatex
	%	\caption{Single UAV downlink system with $K$ UEs and $W$~IRSs.} 
	%	\label{fig:System model}
	%\end{figure}

	\begin{figure} 
		\centering
		\includegraphics[width=0.9 \columnwidth]{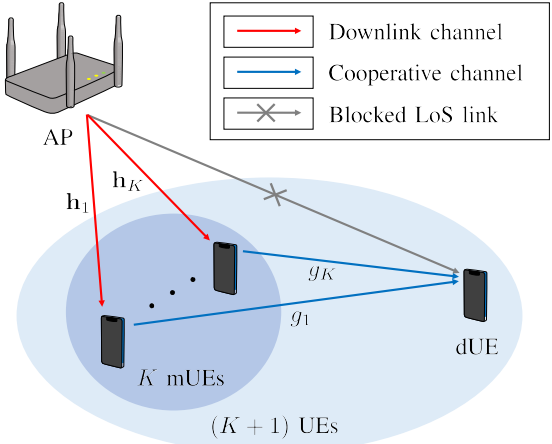}
		\caption{THz MU-MISO downlink system operating in two phases.} 
		\label{fig:System model}
	\end{figure}	
	
	\section{System Model}\label{sec:model}
	In this paper, we consider the MU-MISO downlink system operating in THz frequency bands.
	Because of the reduced multipath effect in the THz frequencies, we consider the system in which only LoS links exist.
	In the system, the AP equipped with $N_t$ antennas serves $(K+1)$ single antenna UEs.\footnote{ The assumption of single antenna UEs is effectively the same as UEs with multiple antennas employing beamforming. We aim to cover a more general case of UEs with multiple antennas supporting multiplexing in our future work.} 
	Among all UEs, the $K$ UEs denoted as the mUEs have the LoS links from the AP, whereas the UE denoted as the dUE experiences blockage of the LoS link from the AP. 
	Since there is no other link to the dUE, it is impossible for the AP to serve every UE with the conventional non-cooperative multiple access techniques such as NOMA, SDMA, and RSMA.
	
	To resolve this issue, we adopt CRS, which enables the AP to communicate with every UE. As illustrated in Fig. \ref{fig:System model}, CRS operates in two phases. 
	The AP transmits messages to the $K$ mUEs in the first phase, and then the $K$ mUEs transmit messages to the dUE in the second phase.
 	To avoid inefficient use of resource while transmitting the message for the dUE, we propose eCRS that extracts the message for the dUE from the common message.
	%We also propose eCRS regarding how the message for the dUE are cooperatively relayed through the $K$ mUEs.
	Furthermore, we investigate two extreme cases of eCRS, namely, IeCRS and DeCRS.
	The details of the framework are explained in the following subsections.
	Throughout the paper, the $K$ mUEs are indexed by a set $\mathcal{K}=\left\{1, \cdots, K\right\}$, and the messages for the $k$-th mUE and the dUE are denoted as $W_k$ and $W_\mathrm{d}$, respectively.
	
	\begin{figure}%	
		\subfloat[Message transmission at the AP-($k$-th mUE) link in the first phase of IeCRS.]{{\includegraphics[width=0.5\textwidth ]{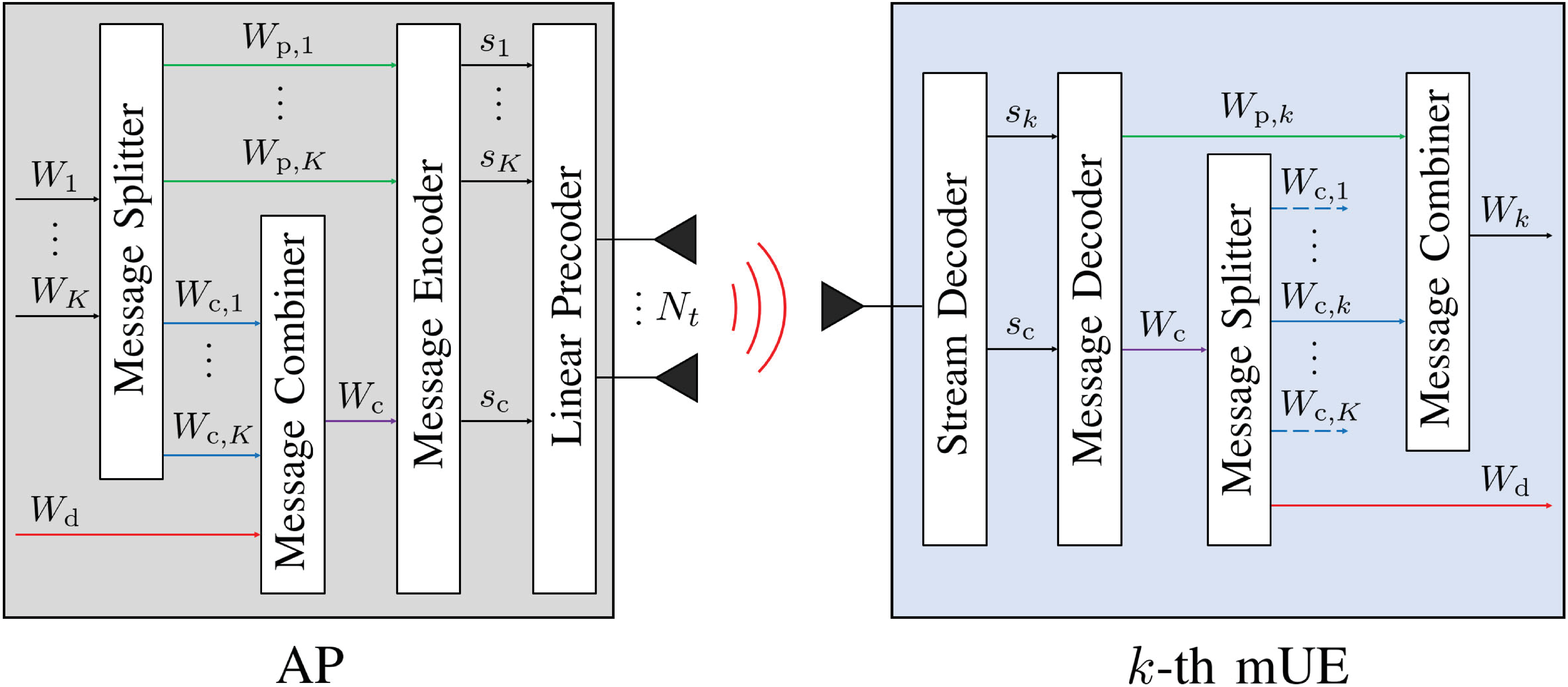}}}% 		
		\hfil
		\subfloat[Message transmission at the ($k$-th mUE)-dUE link in the second phase of IeCRS.]{{\includegraphics[width=0.5\textwidth ]{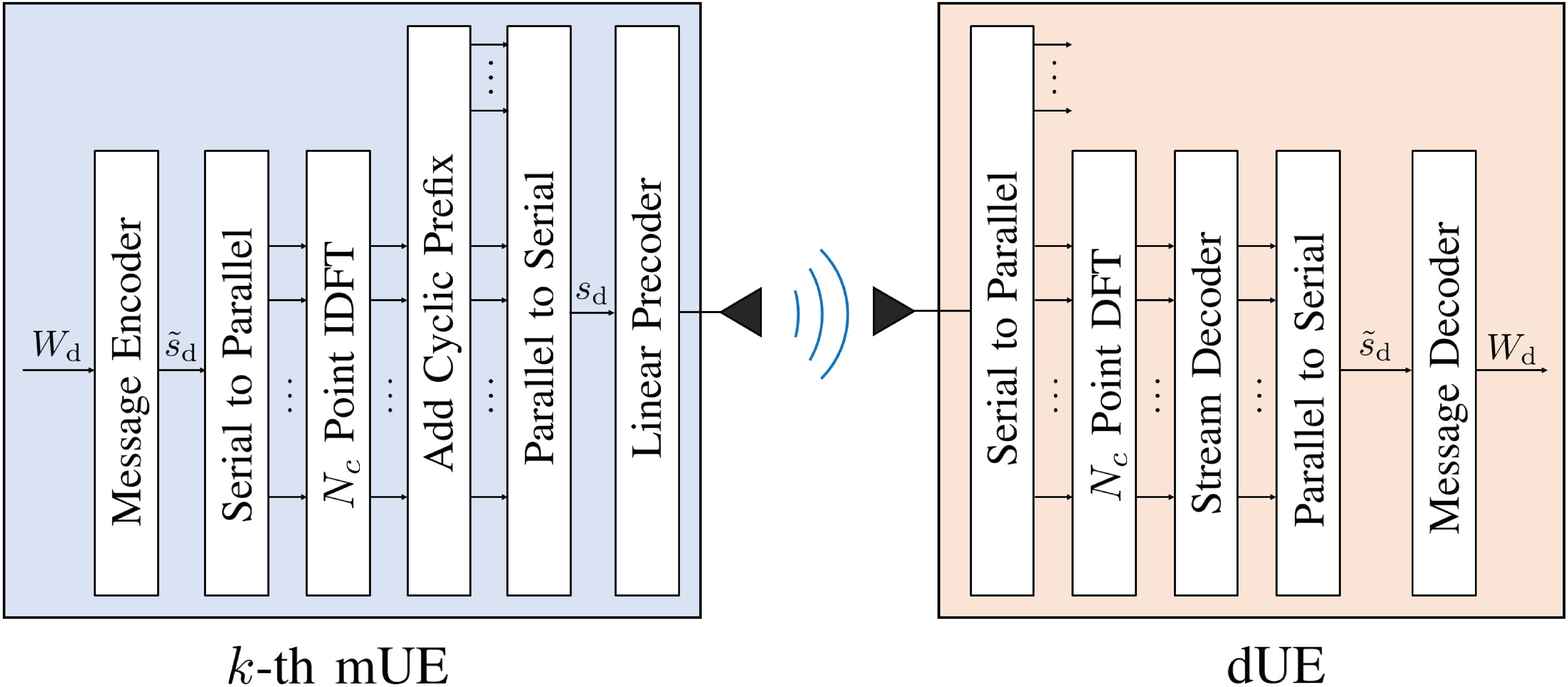} }}% 					
		\caption{Block diagram of the transmission framework of IeCRS.}% 		
		\label{fig:IDT}% 		
	\end{figure}

	\subsection{Transmission Framework of eCRS}
	In this subsection, we describe the transmission framework of eCRS.
	In the first phase, the AP splits the message $W_\mathrm{d}$ with respect to the subsets of the set $\mathcal{K}$, where each subset $\cS$ corresponds to the group of mUEs that decodes the message. The resulting messages are denoted as $\left\{ W_{\mathrm{d}}^{\cS}\right\}_{\cS \subset \cK}$, where $\cS \neq \varnothing$. 
	Similarly, the AP splits the message $W_k$ into private part $W_{\mathrm{p},k}$ and multiple common parts $\left\{ W_{\mathrm{c},k}^{\cS_i}\right\}_{\cS_i \subset \cK}$, where every $\cS_i$ includes the element $k$.
	Next, the AP jointly combines the messages with the same subset $\cS$ to generate the common message $W_{\textrm{c}}^{\cS}$, e.g., for $\cS = \{1,2\}$, the messages $W_{\mathrm{c},1}^{\{1,2\}}, W_{\mathrm{c},2}^{\{1,2\}},$ and $W_{\mathrm{d}}^{\{1,2\}}$ are combined into $W_{\textrm{c}}^{\{1,2\}}$.
	Then, the private part $W_{\mathrm{p},k}$ is encoded into the private stream $s_k$, and the common message $W_{\mathrm{c}}^{\cS}$ is encoded into the common stream $s_{\mathrm{c},\cS}$. 
	Since every mUE in the subset $\cS$ decodes the common stream $s_{\mathrm{c},\cS}$, every mUE must decode multiple common streams.
	While this is possible with multiple SIC layers, our scenario of interest considers the practical case of one SIC layer. 
	In result, we select the subsets so that each mUE only decodes a single common stream.
	The selection of subsets for eCRS with one SIC layer should satisfy the conditions given as
	\begin{align}
		\bigcap \left\{ \cS_i \right\}  = \varnothing, \ \bigcup \left\{ \cS_i \right\} = \cK, \ \cS_i \subset \cK,
	\end{align}
	where $\cS_i$ is the subset, which contains the group of mUEs that decodes the common stream $s_{\mathrm{c},\cS_i}$.
	
	While there can be various approaches to select the best subsets for eCRS with one SIC layer, we focus on two extreme cases, which are IeCRS and DeCRS.
	The subsets of IeCRS and DeCRS are denoted as $\cS_1 = \cK$ and $\left\{\cS_i\right\}_{i=1}^{K} = \left\{ \{1\}, \cdots , \{K\}\right\}$, respectively.
	Motivated by the 1-layer RSMA in \cite{RSMA1layer1,RSMA1layer2}, IeCRS explores the case where the common message is for all the mUEs.
	In contrast to IeCRS, DeCRS is the other extreme case with $K$ common messages, where every mUE has its own common message.
	
	\subsection{IeCRS}
	In IeCRS, the AP selects a single subset as $\cS_1 = \cK$, where the corresponding common message is intended for all mUEs. 
	Hence, the message $W_\mathrm{d}^{\cK}$ is equal to the message $W_\mathrm{d}$, and the message $W_k$ is split into the private part $ W_{\mathrm{p},k}$ and common part $W_{\mathrm{c},k}^{\cK}$.
	For simplicity, we neglect the index $\cK$ in IeCRS. 
	The overall message transmission and reception of IeCRS in the first phase is described in Fig. \ref{fig:IDT} (a).
	With the common stream $s_\mathrm{c}$ and private stream $s_k$, the AP linearly precodes the $\left(K+1\right)$ streams and transmits the signal given as
	\begin{align}
		\bx=\bff_\mathrm{c} s_\mathrm{c}+\sum_{k=1}^K \bff_k s_k,
	\end{align}
	where $\bff_\mathrm{c} \in \mathbb{C}^{N_t \times 1}$ and $\bff_k \in \mathbb{C}^{N_t \times 1}$ are the linear precoders for the common stream $s_\mathrm{c}$ and private stream $s_k$, respectively.
	
	The received signal at the $k$-th mUE during the coherence time block is given as
	\begin{align}
		y_k=\bh^\mathrm{H}_k \bF \bs +z_k, \label{eq: 2}
	\end{align}
	where $\bh_k \in \mathbb{C}^{N_t \times 1}$ is the downlink channel for the AP-($k$-th mUE) link, $\bF \in \mathbb{C}^{N_t \times (K+1)}$ is the precoder matrix defined as $\bF= \left[\bff_\mathrm{c}, \bff_1, \cdots, \bff_K \right]$ with the AP power constraint as $\mathrm{tr}(\bF \bF^{\mathrm{H}}) \leq P_{\mathrm{AP}}$, $\bs$ is the stream vector defined as $\bs=\left[s_\mathrm{c}, s_1, \cdots, s_K\right]^{\mathrm{T}}$ satisfying $\mathbb{E}\left\{ \bs \bs^\mathrm{H}\right\}=\bI_{(K+1)}$, and $z_k \sim \mathcal{CN}\left( 0, N_0 \right)$ is the additive white Gaussian noise (AWGN).
	Without loss of generality, we set the noise variance $N_0$ equal to 1 throughout the paper. 
	
	After receiving the signal, following the same procedure with existing RSMA techniques \cite{RSMA2,RSMA3,RSMA4}, the $k$-th mUE first decodes the common stream $s_\mathrm{c}$, which contains the common message $W_{\mathrm{c}}$. 
	The achievable rate for the common message $W_\mathrm{c}$ at the $k$-th mUE is given as
	\begin{align}
		R_{\mathrm{c}, k} = \log_2 \left( 1+ \frac{|\bh_k^{\mathrm{H}}\bff_\mathrm{c} |^2}{|\bh_k^{\mathrm{H}}\bff_k |^2+ I_k  +1}   \right),
	\end{align}
	where $I_k=\sum_{i \ne k} |\bh_k^{\mathrm{H}}\bff_i|^2$ is the interference from other private streams.
	Since every mUE needs to decode the common message $W_\mathrm{c}$, the corresponding rate of $W_\mathrm{c}$ should not exceed the achievable rate $R_{\mathrm{c}, k}$ respect to every $k \in \mathcal{K}$. The inequality that guarantees the mUEs to successfully decode the common message $W_\mathrm{c}$ is given as
	\begin{align} \label{eq: tradeoff 1}
		C_\mathrm{d}+\sum_{i=1}^K C_i  \le R_{\mathrm{c}, k}, \  \forall k \in \mathcal{K},
	\end{align}
	where $C_k$ is the rate for the common part $W_{\mathrm{c}, k}$, and $C_\mathrm{d}$ is the rate for the message $W_\mathrm{d}$ in the first phase.
	Next, the common stream $s_\mathrm{c}$ is cancelled out from the received signal $y_k$ using the SIC technique.
	Then, the $k$-th mUE decodes the private stream $s_k$, which contains the private part $W_{\mathrm{p},k}$. The achievable rate for the private part $W_{\mathrm{p},k}$ is given as
	\begin{align}
		R_{\mathrm{p},k} = \log_2 \left( 1+ \frac{|\bh_k^{\mathrm{H}}\bff_k |^2}{I_k + 1}   \right).
	\end{align}
	Finally, the $k$-th mUE extracts the common part $W_{\mathrm{c}, k}$ and message $W_\mathrm{d}$ from the common message $W_\mathrm{c}$ and then recombines the private part $W_{\mathrm{p}, k}$ and common part $W_{\mathrm{c}, k}$ into the message $W_k$.
	The achievable rate for the message $W_k$ is given as $R_k=R_{\mathrm{p}, k}+C_k$. 
	
	In the second phase of IeCRS, every mUE transmits the message $W_\mathrm{d}$ to the dUE. The overall process of the message transmission and reception in the second phase is described in Fig. \ref{fig:IDT} (b). 
	The $k$-th mUE encodes the message $W_\mathrm{d}$ into the stream ${s}_\mathrm{d}$ and transmits the precoded signal given as
	\begin{align}
		x_{\mathrm{d},k}= f_{\mathrm{d}, k} s_\mathrm{d},
	\end{align}
	where $f_{\mathrm{d}, k} \in \mathbb{C}$ is a linear precoder at the $k$-th mUE. 
	The received signal at the dUE in the $m$-th time slot is given as
	\begin{align}
		y_\mathrm{d}[m] & = \sum_{k=1}^{K} g_k x_{\mathrm{d}, k}[m-\tau_k] + z_\mathrm{d}[m] \label{eq: 7}\\
		&= \sum_{k=1}^{K} \bar{g}_k 	s_\mathrm{d}[m-\tau_k] + z_\mathrm{d}[m], \label{eq: 8}
	\end{align} 
	where $g_k \in \mathbb{C}$ and $\tau_k \in \mathbb{Z}$ are the channel gain and time delay for the ($k$-th mUE)-dUE link, respectively. The term $z_\mathrm{d}$ is the AWGN, and $\bar{g}_k$ is the effective channel gain of the ($k$-th mUE)-dUE link defined as $\bar{g}_k = g_k f_{\mathrm{d}, k}$. 
	
	Due to the large bandwidth and high sampling rate of THz communication, the signals from multiple mUEs arrive in different taps, making the THz cooperative channel in (\ref{eq: 8}) have the form of a frequency selective channel or single frequency network (SFN) system.
	However, different from the frequency selective channel and conventional SFN \cite{SFN}, the effective channel gain $\bar{g}_k$ can be controlled independently at the $k$-th mUE with the precoder $f_{\mathrm{d}, k}$ to obtain a higher achievable rate.
	We employ an orthogonal-frequency-division-multiplexing (OFDM) approach to make the channel into $N_c$ parallel flat-fading channels.
	After OFDM processing, the received signal at the $n$-th sub-carrier is given as
	\begin{align}
		\tilde{y}_{n}=\tilde{g}_{ n} \tilde{s}_{ n}+\tilde{z}_{n},
	\end{align}
	where $\tilde{s}_{n}$ is the frequency domain stream that satisfies $\mathbb{E}\left\{|\tilde{s}_{n}|^2 \right\}=1$, and  $\tilde{z}_{n}$ is the AWGN.
	The frequency domain channel $\tilde{g}_{n}$ is given as
	\begin{align} \label{eq: ofdm}
		\tilde{g}_{ n}=\sum_{k=1}^{K} \bar{g}_k \exp{\left(\frac{-j2\pi \tau_k n}{N_c}\right)}.
	\end{align}
	We observe that the $N_c$ sub-carrier channels depend on the variable $\bar{g}_k$. By controlling $\bar{g}_k$, i.e., $f_{\mathrm{d},k}$, the quality of the sub-carrier channels can improve.
	By decoding and recombining the messages from $N_c$ received signals, the achievable rate for the message $W_\mathrm{d}$ in the second phase is given as
	\begin{align}
		R_\mathrm{d}^{(2)}&=\frac{1}{N_c+L} \sum_{n=1}^{N_c} \log_2\left( {1+|\tilde{g}_{n}|^2} \right) \nonumber \\
		& \approx \frac{1}{N_c} \sum_{n=1}^{N_c} \log_2\left( {1+|\tilde{g}_{n}|^2} \right) \label{eq: 11},
	\end{align}
	where $L$ is the length of the cyclic prefix. The approximation in (\ref{eq: 11}) is valid since we can set the number of sub-carriers $N_c$ much larger than $L$.
	Finally, the achievable rate for the message $W_\mathrm{d}$ at the dUE is given as
	\begin{align}
		R_\mathrm{d}=\min \left\{C_\mathrm{d}, R_\mathrm{d}^{(2)} \right\},
	\end{align}
	which guarantees the message $W_\mathrm{d}$ to be successfully decoded in both two phases. We design the precoders and message split for IeCRS in Section \ref{sec:IDT}.
	
	\begin{figure}%	
		\subfloat[Message transmission at the AP-($k$-th mUE) link in the first phase of DeCRS.]{{\includegraphics[width=0.5\textwidth ]{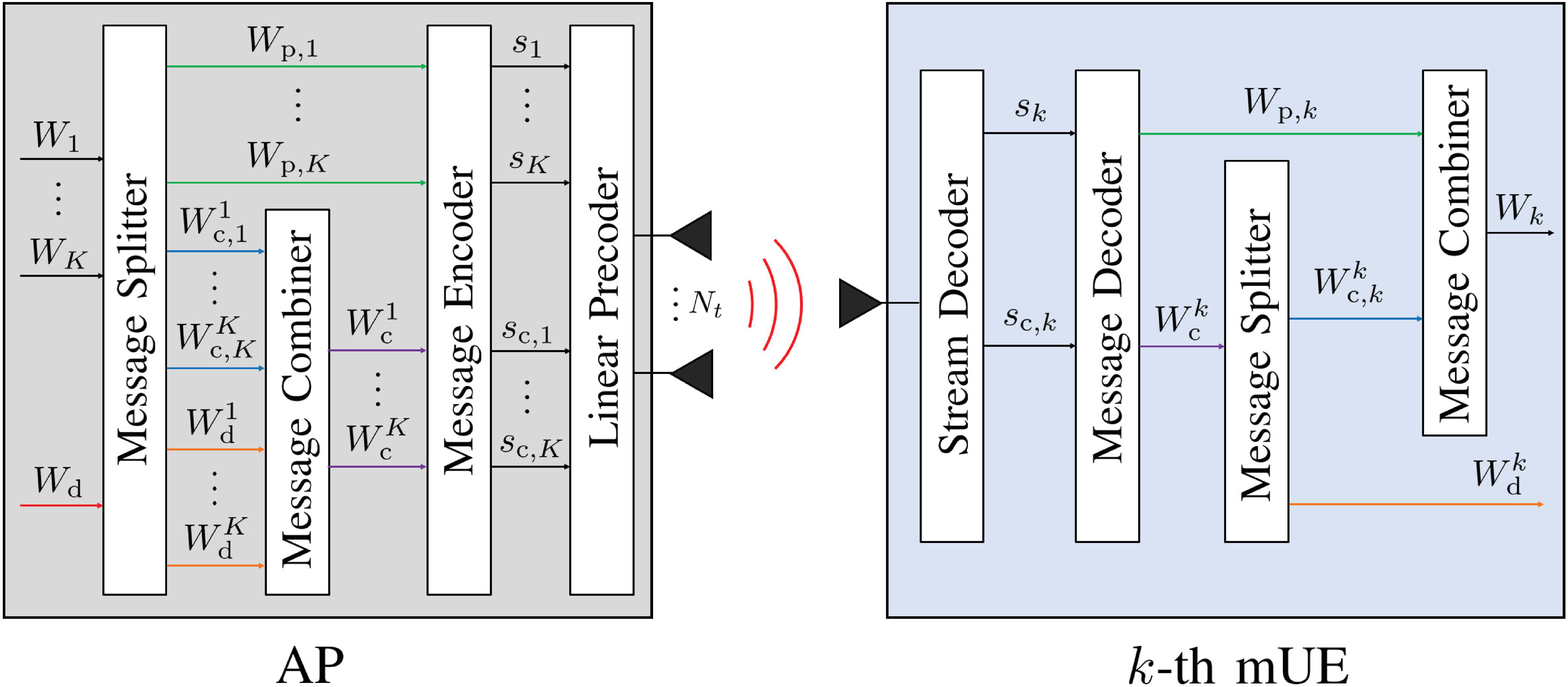} }}% 		
		\hfil
		\subfloat[Message transmission at the mUE-dUE links in the second phase of DeCRS.]{{\includegraphics[width=0.5\textwidth ]{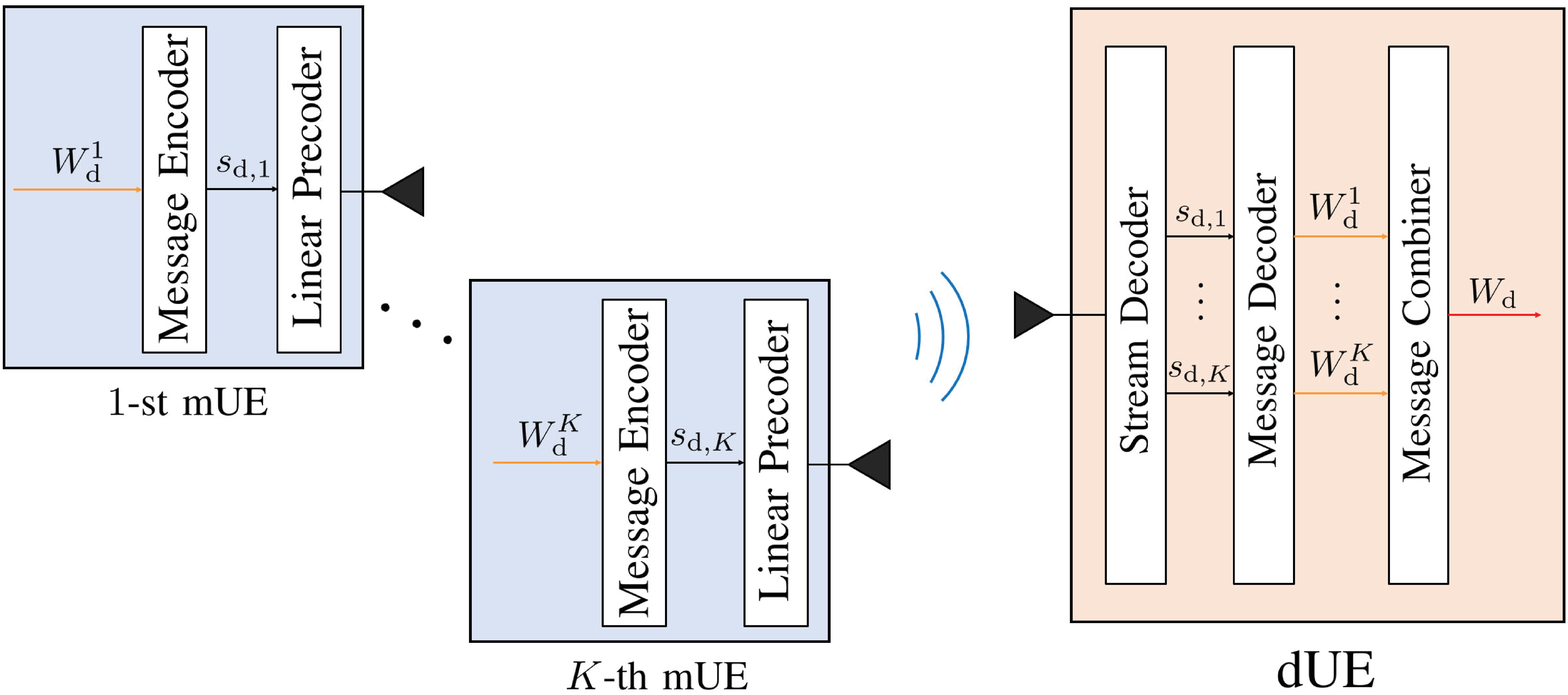} }}% 					
		\caption{Block diagram of the transmission framework of DeCRS.}% 		
		\label{fig:DDT}% 		
	\end{figure}
	
	\subsection{DeCRS}
	For DeCRS, the AP selects the subsets as $\left\{\cS_i\right\}_{i=1}^{{K}} = \left\{ \{1\}, \cdots , \{K\}\right\}$. 
	The overall message transmission and reception of DeCRS in the first phase is described in Fig. \ref{fig:DDT} (a). With $K$ common streams and $K$ private streams, the AP linearly precodes $2K$ streams and transmits the signal given as
	\begin{align}
		\bx = \sum_{k=1}^{K} \bff_{\mathrm{c}, k} s_{\mathrm{c}, k} + \sum_{k=1}^{K} \bff_k s_k, 
	\end{align}
	where ${\bff}_{\mathrm{c}, k} \in \mathbb{C}^{N_t \times 1}$ and $\bff_k \in \mathbb{C}^{N_t \times 1}$ are the linear precoders for the common stream $s_{\mathrm{c}, k}$ and private stream $s_k$, respectively.
	
	The received signal at the $k$-th mUE during the coherence time block is given as
	\begin{align}
		y_k=\bh^\mathrm{H}_k \bF\bs+z_k, \label{eq: 14}
	\end{align}
	where $\bh_k$ is the downlink channel defined in IeCRS, $\bF \in \mathbb{C}^{N_t \times 2K}$ is the precoder matrix defined as $\bF=\left[\bff_{\mathrm{c}, 1}, \cdots\bff_{\mathrm{c}, K}, \bff_1, \cdots, \bff_K \right]$, $\bs$ is the stream vector defined as $\bs= \left[s_{\mathrm{c}, 1}, \cdots,  s_{\mathrm{c}, K}, s_1, \cdots, s_K\right]^{\mathrm{T}}$ satisfying $\mathbb{E}\left\{ \bs \bs^\mathrm{H}\right\}=\bI_{2K}$, and $z_k$ is the AWGN. After receiving the signal, the $k$-th mUE first decodes the common stream $s_{\mathrm{c}, k}$, which contains the common message $W_{\mathrm{c}}^k$. The achievable rate for the common message $W_{\mathrm{c}}^k$ is given as
	\begin{align}
		R_{\mathrm{c},k} = \log_2 \left( 1+ \frac{|\bh_k^{\mathrm{H}}\bff_{\mathrm{c},k} |^2}{|\bh_k^{\mathrm{H}}\bff_{k} |^2+I_k +I_{\mathrm{c},k}+1}\right),	
	\end{align}
	where $I_{\mathrm{c}, k} = \sum_{i \ne k} |\bh_k^{\mathrm{H}}\bff_{\mathrm{c} ,i}|^2$ and $I_k=\sum_{i \ne k} |\bh_k^{\mathrm{H}}\bff_i|^2$ are the interference from other common and private streams, respectively. 
	The inequality that guarantees the $k$-th mUE to successfully decode the common message $W_{\mathrm{c}}^k$ is given as 
	\begin{align} \label{eq: tradeoff 3}
		C_{\mathrm{d}, k} + C_k \le R_{\mathrm{c}, k},
	\end{align}
	where $C_k$ is the rate for the common part $W_{\mathrm{c}, k}^k$, and $C_{\mathrm{d}, k}$ is the rate for the message $W_{\mathrm{d}}^k$ in the first phase. 
	Next, the common stream $s_{\mathrm{c},k}$ is cancelled out from the received signal $y_k$ using the SIC technique. Then, the $k$-th mUE decodes the private stream $s_{k}$, which contains the private part $W_{\mathrm{p},k}$. The achievable rate for the private part $W_{\mathrm{p},k}$ is given as
	\begin{align}
		R_{\mathrm{p},k} = \log_2 \left( 1+ \frac{|\bh_k^{\mathrm{H}}\bff_{k} |^2}{I_k +I_{\mathrm{c}, k} +1}   \right).
	\end{align}
	Finally, the $k$-th mUE extracts the common part $W_{\mathrm{c}, k}^k$ and message $W_{\mathrm{d}}^k$ from the common message $W_{\mathrm{c}}^k$ and then recombines the private part $W_{\mathrm{p}, k}$ and common part $W_{\mathrm{c}, k}^k$ into the message $W_k$. The achievable rate for the message $W_k$ is given as $R_k=R_{\mathrm{p}, k}+C_k$ as in IeCRS.
	
	In the second phase of DeCRS, the $k$-th mUE transmits the distinct message $W_{\mathrm{d}}^k$ to the dUE. 
	The overall process of the message transmission and reception in the second phase is described in Fig. \ref{fig:DDT} (b).
	The $k$-th mUE encodes the message $W_{\mathrm{d}}^k$ into the stream $s_{\mathrm{d}, k}$ and transmits the precoded stream given as
	\begin{align}
		x_{\mathrm{d},k}=f_{\mathrm{d}, k} s_{\mathrm{d},k},
	\end{align}
	where $f_{\mathrm{d}, k} \in \mathbb{C}$ is a linear precoder, and $s_{\mathrm{d}, k}$ satisfies $\mathbb{E}\left\{|s_{\mathrm{d}, k}|^2 \right\}=1$.
	The received signal at the dUE in the $m$-th time slot is given as
	\begin{align}
		y_\mathrm{d}[m] &= \sum_{k=1}^{K} g_{k}  x_{\mathrm{d},k}[m-\tau_k] + z_\mathrm{d}[m] \label{eq: 19} \\
		&= \sum_{k=1}^{K} \bar{g}_{k}  s_{\mathrm{d},k}[m-\tau_k] + z_\mathrm{d}[m] \label{eq: 20},
	\end{align}
	where $g_k$, $\tau_k$, $z_\mathrm{d}$, and $\bar{g}_k$ are defined in IeCRS.
	While the THz cooperative channel in IeCRS has the form of the frequency selective channel due to the identical message $W_\mathrm{d}$, the THz cooperative channel in DeCRS has the form of the uplink multiple access channel, which treats the streams from other mUEs as interference.
	Assuming a conservative scenario that the interference from other mUEs always exists in every time slot, the achievable rate for the message $W_{\mathrm{d}}^k$ in the second phase is given as
	\begin{align}
		R_{\mathrm{d}, k}^{(2)}=\log_2{\left(1+ \frac{|\bar{g}_k|^2}{J_k+1} \right)},
	\end{align}
	where $J_k= \sum_{i\ne k} |\bar{g}_i|^2$ is the interference from other mUEs.
	Finally, the achievable rate for the message $W_\mathrm{d}$ at the dUE is given as
	\begin{align}
		R_\mathrm{d}=\sum_{k=1}^K \min \left\{C_{\mathrm{d}, k}, R_{\mathrm{d}, k}^{(2)}  \right\},
	\end{align}
	which is the summation of the achievable rate for the message $W_{\mathrm{d}}^k$ of every $k \in \mathcal{K}$. The minimum function is used to guarantee the successful decoding of the message $W_{\mathrm{d}}^k$ in both the AP-($k$-th mUE) and ($k$-th mUE)-dUE links.
	We design the precoders and message split for DeCRS in Section \ref{sec:DDT}.
	
	\textit{Remark 1:}
	In IeCRS, all the mUEs must successfully decode the common message in the first phase. 
	If there are mUEs with weak channel conditions, these mUEs may bound the performance for the first phase.
	In contrast, since multiple mUEs transmit the same stream in the second phase, the mUEs collaborate to strengthen the signal in the second phase. 
	In result, the common message for multiple mUEs can be detrimental for the first phase, but can be beneficial for the second phase.
	DeCRS experiences this tradeoff opposite to IeCRS. 
	
	\textit{Remark 2:}
	For IeCRS, the common message affects all the mUEs, since the common stream contains the messages for all $K$ mUEs.
	In contrast, the $k$-th common stream for DeCRS only contains the messages for the dUE and $k$-th mUE, thus, it only affects a single mUE. 
	In result, we can expect that the common message will have more impact on IeCRS than DeCRS.

	\subsection{Channel Model}
	A uniform planar array (UPA) with $N_t$ antennas is implemented at the AP.
	Hence, the channel for the AP-($k$-th mUE) link is given as 
	\begin{align}
		& \bh_k=\sqrt{\beta_0 \left(d_k^{\mathrm{AP-mUE}}\right)^{-\alpha}} \\ \nonumber
		& \times \left[1, \cdots, \exp{\left(j  (N_v-1)\pi \sin \phi_k\right)}\right]^\mathrm{T} \\  \nonumber
		& \otimes \left[1, \cdots, \exp{\left(j (N_h-1) \pi\cos \phi_k \cos \varphi_k\right) }\right]^\mathrm{T},
	\end{align}
	where $\beta_0$ is the path-loss at the unit distance, $d_k^{\mathrm{AP-mUE}}$ is the distance between the AP and $k$-th mUE, and $\alpha$ is the path-loss exponent. The array response vector of the channel is expressed with $\phi_k$ and $\varphi_k$, which are the vertical and horizontal angles between the AP and $k$-th mUE, $N_v$ is the number of vertical antenna elements, and $N_h$ is the number of horizontal antenna elements.
	The channel for the ($k$-th mUE)-dUE link is given as   
	\begin{align}
		g_{k}=\sqrt{\beta_0 \left(d_k^{\mathrm{mUE-dUE}}\right)^{-\alpha}} \exp{\left(j \theta_k\right)},
	\end{align}
	where $d_k^{\mathrm{mUE-dUE}}$ is the distance between the $k$-th mUE and dUE, and $\theta_k$ is the phase of the channel.
	The corresponding time delay is modeled as
	\begin{align}
		\tau_k=\mathrm{round}\left(\frac{f_s d_k^{\mathrm{mUE-dUE}}}{c} \right),
	\end{align}
	where $\mathrm{round}\left(\cdot\right)$ is the round function, $f_s$ is the sampling frequency, and $c$ is the speed of light.

	\section{Identical extraction-based CRS} \label{sec:IDT}
	In this section, we formulate the overall problem for IeCRS and solve it through convex optimization techniques. Also, to compensate for the high complexity of the optimization problem, we derive a transmission strategy with minimal complexity.
	We assume perfect CSI while deriving the precoder and message split. The proposed channel estimation technique is described in Section \ref{sec: est}.
	
	\subsection{Problem Formulation}
	To enlarge the coverage of the system, we focus on maximizing the minimum achievable rate of the mUEs and dUE. The overall problem can be formulated as 
	\begin{align}
		\mathrm{(P1):} \max_{\bF, \bc, \bar{\bg}} &\min \left\{ R_1, \cdots, R_K, R_{\mathrm{d}}\right\} \notag \\
		\text{s.t.} \  &R_k = R_{\mathrm{p},k} + C_k, \ \forall k \in \cK, \tag{1-a} \label{eq: med rate}\\
		&C_{\mathrm{d}} + \sum_{i=1}^K C_i \leq R_{\mathrm{c},k}, \ \forall k \in \cK, \tag{1-b} \label{eq: common rate}\\
		&R_{\mathrm{d}} = \min \left\{C_{\mathrm{d}},R_{\mathrm{d}}^{(2)}\right\}, \tag{1-c} \label{eq: dUE rate} \\
		&\mathrm{tr}(\bF \bF^{\mathrm{H}}) \leq P_{\mathrm{AP}}, \tag{1-d} \label{eq: AP power}\\
		&\lvert \bar{g}_k \rvert ^2 \leq \lvert g_k \rvert^2 P_{k}, \ \forall k \in \cK, \tag{1-e} \label{eq: medium power}
	\end{align}
	where \eqref{eq: AP power} is the power constraint of the AP, and \eqref{eq: medium power} is the power constrant of the $k$-th mUE with power $P_{k}$. The optimization variables $\bc$ and $\bar{\bg}$ are defined as $\bc = \left[ C_1, \cdots, C_K, C_{\mathrm{d}}\right]^{\mathrm{T}}$ and $\bar{\bg} = \left[ \bar{g}_1, \cdots, \bar{g}_K \right]^{\mathrm{T}}$, respectively. Since $\mathrm{(P1)}$ is non-convex, the problem cannot be directly solved. To resolve this issue, we split the problem into two separate problems, each maximizing $C_{\mathrm{d}}$ and $R_{\mathrm{d}}^{(2)}$, and then compute $R_{\mathrm{d}}$ as in \eqref{eq: dUE rate}. The problem for the first phase is formulated as
	\begin{align}
		\mathrm{(P2):} \max_{\bF, \bc} &\min \left\{ R_1, \cdots, R_K, C_{\mathrm{d}}\right\} \notag \\
		\text{s.t.} \  &\eqref{eq: med rate}, \eqref{eq: common rate}, \eqref{eq: AP power}, \notag
	\end{align}
	and the problem for the second phase is formulated as 
	\begin{align}
		\mathrm{(P3):} \max_{\bar{\bg}}  \ &R_{\mathrm{d}}^{(2)} \notag \\
		\text{s.t.} \  &\eqref{eq: medium power}.\notag
	\end{align}
	In the following subsection, we derive the solutions for the problems (P2) and (P3).
	
	\subsection{Proposed Technique}
	We solve (P2) through a weighted minimum mean square error (WMMSE) approach and (P3) through a successive convex approximation (SCA) approach. Due to the high complexity of the SCA approach, we also derive a closed form solution for (P3) under the low SNR assumption.
	
	From (P2), we observe that the problem is similar to the NOUM transmission scenario with RSMA \cite{RSMA4}, where we can interpret $\left\{ W_1 , \cdots, W_K \right\}$ as the unicast messages and $W_{\mathrm{d}}$ as the multicast message. Thus, we adopt the WMMSE approach to the problem (P2) with some adjustments to match our framework. 
	
	The $k$-th mUE first decodes $W_\mathrm{c}$ and removes the corresponding stream $s_\mathrm{c}$ from the received signal $y_k$. Next, the $k$-th mUE decodes the private part $W_{\mathrm{p},k}$. Through the equalizers $w_{\mathrm{c},k}$ and $w_{\mathrm{p},k}$, the $k$-th mUE estimates the streams $\hat{s}_{\mathrm{c},k} = w_{\mathrm{c},k} y_k$ and $\hat{s}_{k} = w_{\mathrm{p},k} \left( y_k - \hat{s}_{\mathrm{c},k} \right)$, respectively. The mean square errors (MSEs) of the $k$-th mUE can be expressed as
	\begin{align}
		\varepsilon_{\mathrm{c},k} &= |w_{\mathrm{c},k}|^2 T_{\mathrm{c},k} - 2 \Re \left\{ w_{\mathrm{c},k} \bh_k^{\mathrm{H}}\bff_{\mathrm{c}}\right\} +1,\\
		\varepsilon_{\mathrm{p},k} &= |w_{\mathrm{p},k}|^2 T_{\mathrm{p},k} - 2 \Re \left\{ w_{\mathrm{p},k} \bh_k^{\mathrm{H}}\bff_k\right\} +1,
	\end{align} 
	where $T_{\mathrm{c},k} = |\bh^{\mathrm{H}}_k \bff_{\mathrm{c}}|^2 + \sum_{i =1}^K |\bh^{\mathrm{H}}_k \bff_i|^2+1$ and $T_{\mathrm{p},k} = T_{\mathrm{c},k} - |\bh^{\mathrm{H}}_k \bff_{\mathrm{c}}|^2$. To obtain the minimum MSE (MMSE) equalizers, we compute $\frac{d\varepsilon_{\mathrm{c},k}}{dw_{\mathrm{c},k}} = 0$ and $\frac{d\varepsilon_{\mathrm{p},k}}{dw_{\mathrm{p},k}} = 0$ to gain
	\begin{align}
		w_{\mathrm{c},k}^{\mathrm{MMSE}} &= \bff_{\mathrm{c}}^{\mathrm{H}} \bh_k T_{\mathrm{c},k}^{-1},\\
		w_{\mathrm{p},k}^{\mathrm{MMSE}} &= \bff_k^{\mathrm{H}} \bh_k T_{\mathrm{p},k}^{-1}.
	\end{align}
	By substituting the MMSE equalizers, the MSEs can be expressed as $\varepsilon_{\mathrm{c},k} = \left( T_{\mathrm{c},k} - |\bh^{\mathrm{H}}_k\bff_{\mathrm{c}}|^2\right)/T_{\mathrm{c},k}$ and $\varepsilon_{\mathrm{p},k} = \left( T_{\mathrm{p},k} - |\bh^{\mathrm{H}}_k\bff_{k}|^2\right)/T_{\mathrm{p},k}$. Then, the achievable rates can be expressed in alternate forms as $R_{\mathrm{c},k} = -\log_2 \left( \varepsilon_{\mathrm{c},k}\right)$ and $R_{\mathrm{p},k} = -\log_2 \left( \varepsilon_{\mathrm{p},k}\right)$. Furthermore, we define the weighted MSEs (WMSEs) as $\xi_{\mathrm{c},k} = \mu_{\mathrm{c},k} \varepsilon_{\mathrm{c},k} - \ln (\mu_{\mathrm{c},k})$ and $\xi_{\mathrm{p},k} = \mu_{\mathrm{p},k} \varepsilon_{\mathrm{p},k} - \ln (\mu_{\mathrm{p},k})$, where $\mu_{\mathrm{c},k}$ and $\mu_{\mathrm{p},k}$ are the weight variables. In result, the achievable rates can be expressed through the WMMSEs as
	\begin{align}
		\xi^{\mathrm{MMSE}}_{\mathrm{c},k} &= \min_{w_{\mathrm{c},k}, \mu_{\mathrm{c},k}} \xi_{\mathrm{c},k}  = 1 - R_{\mathrm{c},k} \ln (2),\\
		\xi^{\mathrm{MMSE}}_{\mathrm{p},k} &= \min_{w_{\mathrm{p},k}, \mu_{\mathrm{p},k}} \xi_{\mathrm{p},k}  = 1 - R_{\mathrm{p},k} \ln (2),
	\end{align}
	where the optimal weights are derived as $\mu^{\mathrm{WMMSE}}_{\mathrm{c},k} = \left(\varepsilon_{\mathrm{c},k}^{\mathrm{MMSE}}\right)^{-1}$ and $\mu^{\mathrm{WMMSE}}_{\mathrm{p},k} = \left(\varepsilon_{\mathrm{p},k}^{\mathrm{MMSE}}\right)^{-1}$.
	
	Through the relationship between the achievable rate and WMMSE, the problem (P2) is transformed by interchanging the achievable rates into the WMMSE forms as
\begin{align}
\mathrm{(P2}&\mathrm{.1):} \min_{\cV_{\mathrm{IeCRS}}} \  t_0  \notag \\
& \ \text{s.t.} \  X_k + (\xi_{\mathrm{p},k} - 1) / \ln(2) \leq t_0, \ \forall k \in \cK, \tag{2.1-a}\\
&\qquad (\xi_{\mathrm{c},k}-1) / \ln (2) \leq X_{\mathrm{d}} + \sum_{i=1}^K X_i , \ \forall k \in \cK, \tag{2.1-b}\\
&\qquad X_{\mathrm{d}} \leq t_0, \tag{2.1-c} \\
&\qquad \eqref{eq: AP power}, \notag
\end{align}
	where $\bx = [X_1, \cdots, X_K, X_{\mathrm{d}}]^{\mathrm{T}} =  -\bc$, and $t_0$ is a slack variable to express $\min \left\{ R_1, \cdots, R_K, C_{\mathrm{d}}\right\}$. The variables of the optimization problem are defined as the set $\cV_{\mathrm{IeCRS}} = \left\{ \bF, \bx, \bw, \bmu, t_0 \right\}$ with $\bw = \left[ w_{\mathrm{p},1} , \cdots, w_{\mathrm{p},K}, w_{\mathrm{c},1}, \cdots, w_{\mathrm{c},K}\right]^{\mathrm{T}}$ and $\bmu = \left[ \mu_{\mathrm{p},1} , \cdots, \mu_{\mathrm{p},K}, \mu_{\mathrm{c},1}, \cdots, \mu_{\mathrm{c},K}\right]^{\mathrm{T}}$. By solving (P2.1), the resulting minimum rate of the system is $\min \left\{ R_1, \cdots, R_K, C_{\mathrm{d}}\right\} = -t_0$. While (P2.1) is non-convex in general, the problem is convex with respect to each variable set $\left\{ \bF, \bx, t_0 \right\}$ and $\left\{ \bw, \bmu \right\}$ by fixing the other variable set.
	Thus, the problem can be efficiently solved iteratively through alternating optimization (AO). 
	
	To solve the non-convex problem (P3), we introduce slack variables to adopt the SCA approach. The problem can be transformed as
\begin{align}
\mathrm{(P3.1):} \max_{\bar{\bg}, \ba, \bb, \bu} \  &\frac{1}{N_c}\sum_{n=1}^{N_c}  \log_2 (1 + u_n) \notag \\
\text{s.t.} \  &u_n \leq a_n^2 + b_n^2, \ \forall n \in \cN, \tag{3.1-a} \label{eq: slack}\\
&\ba = \Re \{\tilde{\bg}\}, \bb = \Im \{\tilde{\bg}\}, \  \forall n \in \cN, \tag{3.1-b} \label{eq: for g}\\
&\eqref{eq: medium power}, \notag
\end{align}
	where $\tilde{\bg} = [\tilde{g}_1, \cdots ,\tilde{g}_{N_c}]^{\mathrm{T}}$.
	The slack variables are defined as $\ba = \left[a_1 , \cdots, a_{N_c}\right]^{\mathrm{T}},\bb = \left[b_1 , \cdots, b_{N_c}\right]^{\mathrm{T}}$, and $\bu = \left[u_1 , \cdots, u_{N_c}\right]^{\mathrm{T}}$. The sub-carriers are indexed by a set $\cN =\left\{1, \cdots, N_c\right\}$. To maximize the objective, $u_n$ will be maximized until \eqref{eq: slack} is met with equality, resulting in $u_n =  \lvert \tilde{g}_n \rvert^2$. Thus, the problems (P3) and (P3.1) are indeed equivalent. While (P3.1) is non-convex due to \eqref{eq: slack}, the right-hand side of the constraint is convex, motivating us to use the SCA approach. In result, the surrogate optimization problem for the $\ell$-th iteration will be
\begin{align}
\mathrm{(P3.2):} \max_{\bar{\bg}, \ba, \bb, \bu} \  &R_{\mathrm{d}}^{(2)} = \frac{1}{N_c}\sum_{n=1}^{N_c}  \log_2 (1 + u_n) \notag \\
\text{s.t.} \  &u_n \leq \tilde{a}^{(\ell)}_n + \tilde{b}^{(\ell)}_n , \ \forall n \in \cN, \tag{3.2-a}\\
&\eqref{eq: for g}, \eqref{eq: medium power}, \notag
\end{align}
	where $\tilde{a}^{(\ell)}_n$ and $\tilde{b}^{(\ell)}_n$ are the first-order derivatives of the right-hand side of \eqref{eq: slack} defined as
	\begin{align}
		\tilde{a}^{(\ell)}_n &= \left(a_n^{(\ell)}\right)^2 + 2 \left\{ a_n^{(\ell)} \left(a_n-a_n^{(\ell)}\right)\right\}, \label{eq: bound a}\\
		\tilde{b}^{(\ell)}_n &= \left(b_n^{(\ell)}\right)^2 + 2 \left\{ b_n^{(\ell)} \left(b_n-b_n^{(\ell)} \right)\right\}, \label{eq: bound b}
	\end{align}
	with the variables $a_n^{(\ell)}$ and $b_n^{(\ell)}$ defined as the local points for the $\ell$-th iteration. By iteratively solving (P3.2) and updating $a_n^{(\ell)}$ and $b_n^{(\ell)}$ with the solutions from the $(\ell-1)$-th iteration, the solution will converge to a local optimum of (P3.1) \cite{SCA}. By taking the minimum of $C_{\mathrm{d}}$ and $R_{\mathrm{d}}^{(2)}$, we obtain the achievable rate $R_{\mathrm{d}}$. The overall algorithm for IeCRS is shown in Algorithm 1.
	
	\begin{algorithm}[t]
		\begin{algorithmic} [1]
			\caption{Pseudo code for minimum rate maximization in IeCRS}
			\State \textbf{Initialization:} Set $\ell_1 = 0, \bw^{(\ell_1)},$ and $\bmu^{(\ell_1)}$.
			\Repeat
			\State Set $\ell_1 = \ell_1 +1$.
			\State Solve (P2.1) with fixed $\bw$ and $\bmu$.
			\State Update $\bw^{(\ell_1)}$ and $\bmu^{(\ell_1)}$.
			\Until \ $t_0$ decreases by a fraction below a predefined threshold.
			\State $C_{\mathrm{d}} = -t_0$.
			\State Set $\ell_2 = 0, \ba^{(\ell_2)},$ and $\bb^{(\ell_2)}$.
			\Repeat
			\State Set $\ell_2 = \ell_2 +1$.
			\State Solve (P3.2).
			\State Update $\ba^{(\ell_2)}$ and $\bb^{(\ell_2)}$ as the solutions of (P3.2).
			\Until \ $R_{\mathrm{d}}^{(2)}$ increases by a fraction below a predefined threshold.
			\State $R_{\mathrm{d}} = \min \left\{ C_{\mathrm{d}},R_{\mathrm{d}}^{(2)}\right\}$.
		\end{algorithmic}
	\end{algorithm} 
	
	\subsection{Low Complexity Approach}
	While the convex optimization approach is effective in performance, the complexity of iteratively solving (P3.2) is quite high, where the slack variables have the size of the number of OFDM sub-carriers. Due to the excessive use of bandwidth in THz frequencies, the number of sub-carriers are expected to be huge. To compensate for this factor, we propose a low complexity approach to solve (P3).
	
	Due to the low transmit power of the mUEs, we assume that the dUE operates in the low SNR regime. The achievable rate $R_{\mathrm{d}}^{(2)}$ can then be simplified as
	\begin{align}
		\sum_{n=1}^{N_c}  \log_2 (1 + \lvert \tilde{g}_n \rvert^2) \approx \sum_{n=1}^{N_c}   \lvert \tilde{g}_n \rvert^2 / \ln(2).
	\end{align}
	We also express $\tilde{g}_n$ in an alternate form from \eqref{eq: ofdm} as $\tilde{g}_n = \bOmega_{n}^{\mathrm{H}} \bar{\bg}$, where $\bOmega_n = \left[ \Omega_{n,1}, \cdots, \Omega_{n,K} \right]^{\mathrm{T}}$, with $\Omega_{n,k} = \exp\left(\frac{j2\pi \tau_k n}{N_c}\right)$. 
	After neglecting $\ln(2)$ for simplicity, the achievable rate can be simplified as
	\begin{align}
		\sum_{n=1}^{N_c}  \lvert \tilde{g}_n \rvert^2 &= \sum_{n=1}^{N_c} \bar{\bg}^{\mathrm{H}} \bOmega_{n} \bOmega_{n}^{\mathrm{H}} \bar{\bg} \\
		&= \bar{\bg}^{\mathrm{H}} \bOmega \bar{\bg}, \notag
	\end{align}
	where $\bOmega = \sum_{n=1}^{N_c} \bOmega_n \bOmega_n^{\mathrm{H}}$. Before proposing our low complexity approach, we first state the following lemma.
	
	\begin{lemma} \label{lemma 1}
		For any number of sub-carriers $N_c$,
		\begin{align}
			\bOmega = N_c \bI_K.
		\end{align}
	\end{lemma}
	\begin{proof}
		For $a \ne b , \{a, b\} \in \cK $, the $(a,b)$-th element of $\bOmega$ can be derived as
		\begin{align}
			\Omega[a,b] &= \sum_{n=1}^{N_c} \exp \left( \frac{j2\pi (\tau_a - \tau_b) n }{N_c}\right) \\
			&= \sum_{n=1}^{N_c} \omega_{a,b}^n = \omega_{a,b} \frac{1- \omega_{a,b}^{N_c}}{1-\omega_{a,b}} \notag\\
			&= 0, \notag
		\end{align}
		where $\omega_{a,b} = \exp \left( \frac{j 2 \pi (\tau_a - \tau_b)}{N_c}\right)$. The $a$-th diagonal element of $\bOmega$ can be derived as $N_c$ in the same way, which finishes the proof.
	\end{proof}
	Using Lemma 1, (P3) can be expressed as
\begin{align}
\mathrm{(P3.3):} \max_{\bar{\bg}} &\frac{1}{N_c \ln(2)} \norm{\bar{\bg}}^2 \notag \\
\text{s.t.} \ &\eqref{eq: medium power}. \notag
\end{align}
	From (P3.3), we observe that the optimal communication strategy in the low SNR regime is to simply use all the power of the mUEs, where their phase values are irrelevant.\footnote{In the rare occasion when the delays of multiple mUEs overlap, the optimal strategy is derived in a similar manner. In result, the mUEs have to use all its power, and the received signals from the mUEs with the same delays must have equal phase values.} This strategy results in a drastic reduction of complexity, adequate for scenarios with huge numbers of sub-carriers.
	
	\section{Distinct extraction-based CRS} \label{sec:DDT}
	\subsection{Problem Formulation}
	Similar to IeCRS, we aim to maximize the minimum achievable rate of the system. The overall problem is formulated as
\begin{align}
\mathrm{(P4):} \max_{\bF, \bc, \bar{\bg}} &\min \left\{ R_1, \cdots, R_K, R_{\mathrm{d}}\right\} \notag \\
\text{s.t.} \ &C_{\mathrm{d},k} + C_k \leq R_{\mathrm{c},k}, \ \forall k \in \cK, \tag{4-a} \\
&R_{\mathrm{d}} = \sum_{k=1}^K \min \left\{C_{\mathrm{d},k},R_{\mathrm{d},k}^{(2)}\right\}, \tag{4-b} \label{eq: DDT rate} \\
&\eqref{eq: med rate}, \eqref{eq: AP power}, \eqref{eq: medium power}, \notag
\end{align}
	where $\bc = \left[ C_1, \cdots, C_K , C_{\mathrm{d},1} , \cdots, C_{\mathrm{d},K}\right]$ for DeCRS.
	\subsection{Proposed Technique}
	In this subsection, we solve (P4) through the WMMSE approach. We assume that the dUE uses the equalizer $w_{\mathrm{d},k}$ to estimate the stream $\hat{s}_{\mathrm{d},k} = w_{\mathrm{d},k} y_{\mathrm{d}}$. Similar to IeCRS, we define the MSEs of the $k$-th mUE as 
	\begin{align}
		\varepsilon_{\mathrm{c},k} &= |w_{\mathrm{c},k}|^2 T_{\mathrm{c},k} - 2 \Re \left\{ w_{\mathrm{c},k} \bh_k^{\mathrm{H}}\bff_{\mathrm{c},k}\right\} +1,\\
		\varepsilon_{\mathrm{p},k} &= |w_{\mathrm{p},k}|^2 T_{\mathrm{p},k} - 2 \Re \left\{ w_{\mathrm{p},k} \bh_k^{\mathrm{H}}\bff_k\right\} +1,
	\end{align}
	where $T_{\mathrm{c},k} =\sum_{i'=1}^K |\bh^{\mathrm{H}}_k \bff_{\mathrm{c},i'}|^2 + \sum_{i=1}^K |\bh^{\mathrm{H}}_k \bff_i|^2+1$ and $T_{\mathrm{p},k} = T_{\mathrm{c},k} - |\bh^{\mathrm{H}}_k \bff_{\mathrm{c},k}|^2$. We define the MSE of the stream for the dUE from the $k$-th mUE as 
	\begin{align}
		\varepsilon_{\mathrm{d},k} &= |w_{\mathrm{d},k}|^2 T_{\mathrm{d},k} - 2 \Re \left\{ w_{\mathrm{d},k} \bar{g}_k \right\} +1,
	\end{align}
	where $T_{\mathrm{d},k} = \sum_{i=1}^K \lvert \bar{g}_i \rvert^2 +1$. The definitions of the WMSEs and weights are neglected due to redundancy. 
	
	In result, (P4) can be transformed as 
\begin{align}
\mathrm{(P4.}&\mathrm{1):} \min_{\cV_{\mathrm{DeCRS}}} \  t_0  \notag \\
&\text{s.t.}  \ X_k + (\xi_{\mathrm{p},k} - 1) / \ln(2) \leq t_0, \ \forall k \in \cK, \tag{4.1-a}\\
& \qquad (\xi_{\mathrm{c},k}-1) / \ln (2) \leq X_{\mathrm{d},k} + X_k , \ \forall k \in \cK, \tag{4.1-b}\\
& \qquad X_{\mathrm{d},k} \leq t_k, \ \forall k \in \cK, \tag{4.1-c} \label{eq: DDT phase1}\\
& \qquad (\xi_{\mathrm{d},k}-1) / \ln(2) \leq t_k, \ \forall k \in \cK, \tag{4.1-d} \label{eq: DDT phase2} \\
& \qquad \sum_{k=1}^{K} t_k \leq t_0, \tag{4.1-e} \label{eq: DDT total} \\
& \qquad \eqref{eq: AP power}, \eqref{eq: medium power}, \notag
\end{align}
	where the set of optimization variables is defined as $\cV_{\mathrm{DeCRS}} = \left\{ \bF, \bx, \bw, \bmu, t_0, \bt \right\}$, with $\bx = -\bc$ and $\bt = \left[t_1 , \cdots, t_K\right]^{\mathrm{T}}$. By introducing the slack variables $\{t_k\}$ for every $k \in \cK$, the constraint \eqref{eq: DDT rate} is expressed by the constraints \eqref{eq: DDT phase1}-\eqref{eq: DDT total}, where \eqref{eq: DDT phase1} limits the rate of $C_{\mathrm{d},k}$, \eqref{eq: DDT phase2} limits the rate of $R_{\mathrm{d},k}^{(2)}$, and \eqref{eq: DDT total} sums up the rates of the mUEs. 
	Through the AO approach, (P4.1) is solved to reach a local optimum. The overall algorithm for DeCRS is shown in Algorithm 2.
	
	\begin{algorithm}[t]
		\begin{algorithmic} [1]
			\caption{Pseudo code for minimum rate maximization in DeCRS}
			\State \textbf{Initialization:} Set $\ell = 0, \bw^{(\ell)},$ and $\bmu^{(\ell)}$.
			\Repeat
			\State Set $\ell = \ell+1$.
			\State Solve (P4.1) with fixed $\bw$ and $\bmu$.
			\State Update $\bw^{(\ell)}$ and $\bmu^{(\ell)}$.
			\Until \ $t_0$ decreases by a fraction below a predefined threshold.
			\State $R_{\mathrm{d}} = -t_0$.
		\end{algorithmic}
	\end{algorithm} 
	
	\section{Channel Estimation} \label{sec: est}
	In the transmission framework of eCRS, the AP requires the CSI of both the downlink and THz cooperative channels to find the optimal beamformers for the system.
	In the reception process, the $k$-th mUE requires the CSI of the AP-($k$-th mUE) link, and the dUE requires the CSI of every link in the THz cooperative channel to decode the intended message.
	We exploit the channel reciprocity by employing the time division duplexing (TDD). 
	The channel of the AP-($k$-th mUE) link, which appears in (\ref{eq: 2}) and (\ref{eq: 14}) can be estimated by the conventional MU-MISO downlink channel estimation techniques \cite{CH1, CH2, CH3}. For example, the $K$ mUEs transmit orthogonal pilot signals to the AP, and the AP estimates the reciprocal channel of the AP-($k$-th mUE) link from the received pilot signal. We omit the detailed process of the channel estimation for the downlink channel since it is well described in other literatures.   
	
	We focus on the channel estimation for the ($k$-th mUE)-dUE link, which appears in (\ref{eq: 7}) and (\ref{eq: 19}). Although the channel has the form of the frequency selective channel, the conventional estimation technique is not suitable since the AP requires not only the channel gain but also the time delay of the specific mUE. 
	Hence, we propose a novel channel estimation technique for the THz cooperative channel, which first estimates the time delay $\tau_k$ and then the channel gain $g_k$.
	
	In prior, the $k$-th mUE and dUE share the information about the length $N_p$ pilot given as
	\begin{align}
		\boldsymbol{\psi}_k=\left[\psi_k[0], \cdots, \psi_k[N_p-1]\right]^\mathrm{T},
	\end{align}
	which satisfies $\lVert\boldsymbol{\psi}_k\rVert^2=N_p$.
	Thus, the $k$-th mUE transmits the signal during $N_p$ time slots given as
	\begin{align}
		x_{\mathrm{d}, k}[m]=\sqrt{P} {\psi}_k[m], \ m=0, \cdots, N_p-1,
	\end{align}
	where $P$ is the transmit power of the pilot signal. 
	The dUE receives the pilot signals through the THz cooperative channel during $N_r$ time slots. The length $N_r$ should satisfy the inequality $\tau_\mathrm{max}+N_p \le N_r$ so that every pilot signal is received at the dUE, where  $\tau_{\mathrm{max}}$ is defined as $\tau_{\mathrm{max}}=\max_k{\tau_k}$.
	The received signal at the dUE in the $m$-th time slot is given as
	\begin{align}
		y_\mathrm{d}[m] &= \sum_{k=1}^{K} g_{k} x_{\mathrm{d}, k}[m-\tau_k] + z_\mathrm{d}[m] \nonumber  \\ 
		&= \sum_{k=1}^{K} g_{k}  \sqrt{P} \psi_k[m-\tau_k] + z_\mathrm{d}[m].
	\end{align}
	The received signals during $N_r$ time slots can be reformulated as
	\begin{align}
		\by=\sqrt{P} \boldsymbol{\Psi} \bg + \bz,
	\end{align}
	where $\by=\left[y_{\mathrm{d}}[0], \cdots , y_{\mathrm{d}}[N_r-1] \right]^{\mathrm{T}}$, $\bg=\left[g_1, \cdots,  g_K\right]^\mathrm{T}$, and $\bz=\left[z_{\mathrm{d}}[0], \cdots,z_{\mathrm{d}}[N_r-1] \right]$. The matrix $\boldsymbol{\Psi}$ is defined as 
	$\boldsymbol{\Psi}=\left[{\boldsymbol{\psi}_{1}^{(\tau_{1})}}, \cdots, {\boldsymbol{\psi}_{K}^{(\tau_{K})}} \right]$
	where $\boldsymbol{\psi}_k^{(\tau)}$ is a $\tau$-shifted vector of the $\boldsymbol{\psi}_k$  defined as $\boldsymbol{\psi}_k^{(\tau)}=\left[\boldsymbol{0}_{\tau}^\mathrm{T}, \boldsymbol{\psi}_k^\mathrm{T}, \boldsymbol{0}_{(N_r-N_p-\tau)}^\mathrm{T}  \right]^\mathrm{T}$.
	
	\subsection{Time Delay Estimation}
	
	The dUE estimates the time delay $\tau_k$ and then regenerates the $k$-th column of the matrix $\boldsymbol{\Psi}$ by exploiting the estimated time delay $\hat{\tau}_k$ and pilot $\boldsymbol{\psi}_k$. To estimate the time delay $\tau_k$, the dUE projects the $\tau$-shifted vector ${\boldsymbol{\psi}_k^{(\tau)}}$ onto the received vector $\by$ such as
	\begin{align}
		r_k(\tau)&=\sqrt{P} \left\{  \underbrace{\left({\boldsymbol{\psi}_k^{(\tau)}}\right)^\mathrm{H}{\boldsymbol{\psi}_{k}^{(\tau_{k})}} g_k}_{\text{Signal term}} 
		+ \underbrace{\sum_{i \ne k} \left({\boldsymbol{\psi}_k^{(\tau)}}\right)^\mathrm{H} {\boldsymbol{\psi}_{i}^{(\tau_{i})}} g_i}_{\text{Interference term}}\right\} \nonumber \\ 
		&+\underbrace{\left({\boldsymbol{\psi}_k^{(\tau)}}\right)^\mathrm{H}\bz}_{\text{Noise term}}, \label{eq: 45}
	\end{align}
	where $\tau$ ranges from $0$ to $\tau_{\mathrm{max}}$.
	We propose a maximum projection (MP) estimator given as
	\begin{align}
		\hat{\tau}_k=\argmax_{\tau} \lvert r_k(\tau)\rvert, 
	\end{align}
	which searches for $\tau$ such that the magnitude of the projected value is maximized. 
	To fully utilize the MP estimator, we implement the pilot that has a pseudo-noise property, where the auto-correlation is given as
	\begin{align}
		\mathrm{R}_k(\tau_1-\tau_2)=\left| \left( {\boldsymbol{\psi}_{k}^{(\tau_1)}}\right)^{\mathrm{H}}  {\boldsymbol{\psi}_{k}^{(\tau_2)}} \right| \ll \norm{\boldsymbol{\psi}_{k}}^2, \ \forall \tau_1\ne \tau_2. \label{eq: 47}
	\end{align}
	We also assume that the pilot has the cross-correlation given as
	\begin{align}
		\mathrm{R}_{k, k'}(\tau_1, \tau_2)=\left| \left( {\boldsymbol{\psi}_{k}^{(\tau_1)}}\right)^{\mathrm{H}}  {\boldsymbol{\psi}_{k'}^{(\tau_2)}}  \right| \approx 0, \ \forall k \ne k',    \label{eq: 48}
	\end{align}
	which can suppress the interference term in (\ref{eq: 45}). 
	With the idealistic properties of the pilot, the magnitude of the projected value $r_k(\tau)$ is given as \begin{align}
		\lvert r_k(\tau) \rvert 
		&\approx\left| \sqrt{P} \left({\boldsymbol{\psi}_k^{(\tau)}}\right)^\mathrm{H}{\boldsymbol{\psi}_{k}^{(\tau_{k})}} g_k +\left({\boldsymbol{\psi}_k^{(\tau)}}\right)^\mathrm{H}\bz \right| \\
		& \ll \lvert r_k(\tau_k) \rvert, \ \forall \tau \ne \tau_k.  \notag
	\end{align}
	Hence, the MP estimator can estimate the exact time delay $\tau_k$ with the idealistic properties.
	In this paper, we implement the Zadoff-chu sequence in \cite{zadoff} for the pilot, where the idealistic properties hold when the sequence length is sufficiently long.
	
	\subsection{Channel Gain Estimation}
	From the estimated time delay $\hat{\tau}_k$ and pilot $\boldsymbol{\psi}_k$, the dUE regenerates the matrix $\boldsymbol{\Psi}$ such as
	\begin{align} \label{eq: 51}
		\hat{\boldsymbol{\Psi}}=\left[{\boldsymbol{\psi}_{1}^{(\hat{\tau}_{1})}}, \cdots, {\boldsymbol{\psi}_{K}^{(\hat{\tau}_{K})}} \right].
	\end{align}
	We implement the least-square (LS) estimation technique to estimate the channel such as
	\begin{align}
		\hat{\bg}=\left(\hat{\boldsymbol{\Psi}}^\mathrm{H}\hat{\boldsymbol{\Psi}}\right)^{-1}\hat{\boldsymbol{\Psi}}^\mathrm{H}\by.
	\end{align} 
	In our scenario of interest, the dUE cannot feedback the CSI to the AP through the THz frequency bands since the LoS link between the AP and dUE is blocked. In practice, mobile devices may be able to utilize multiple frequency bands. Hence, we assume that the dUE feeds back the CSI of THz cooperative channel through lower frequency bands, where only a limited amount of the message transmission is available.

	\section{Simulation Results} \label{sec: simul}
	In this section, we verify the performances of the channel estimation technique and the two cases of eCRS with one SIC layer. For the simulations, the carrier frequency and bandwidth are fixed as $f_c = 0.3$ THz and $B = 1 $ GHz, respectively, and the noise power spectral density is fixed as $N_0 = -174$ dBm/Hz. 
	Unless stated otherwise, the power of the AP is assumed as $P_{\mathrm{AP}} = 20$ dBm, and the power of the mUEs is fixed as $P_k = 0$ dBm. The AP and dUE are located at $[0,4,1]$ m and $[8,4,0]$ m, respectively, and the mUEs are uniformly spread in a box with diagonal coordinates $[2,0,0]$ m and $[6,8,0]$ m. 
	
	\begin{figure}%	
		\subfloat[The auto-correlation and cross-correlation of the Zadoff-chu sequence.]{{\includegraphics[width=0.5\textwidth ]{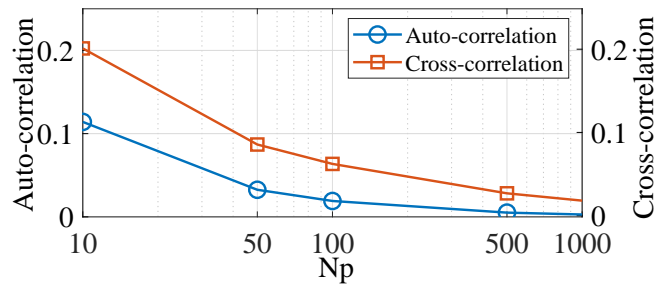}}}% 		
		\hfil
		\subfloat[The DER and NMSE performances.]{{\includegraphics[width=0.5\textwidth ]{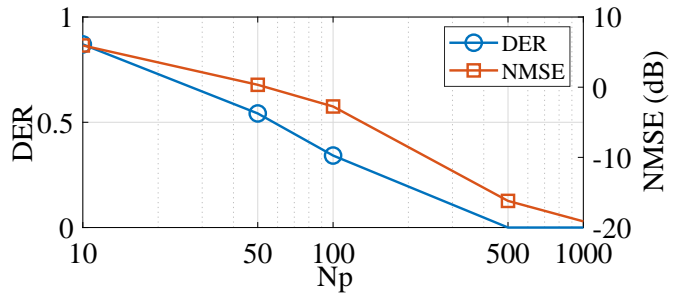}}}%		
		\caption{Performances of the proposed channel estimation technique with respect to the pilot length  $N_p$, where $K=5$ and $N_t=16$.}% 		
		\label{fig: Np}% 		
	\end{figure}
	
	The performance of the channel estimation technique is verified with the delay error rate (DER) of the time delay and the normalized mean square error (NMSE) of the channel gain, which are given as
	\begin{align}
		\mathrm{DER}= \frac{1}{K} \sum^K_{k=1}  \mathrm{Pr}\left( \hat{\tau}_k \ne \tau_k\right),
	\end{align}
	\begin{align}
		\mathrm{NMSE}=\mathbb{E} \left[\frac{ \norm{\hat{\bg}-\bg}^2}{\norm{\bg}^2} \right],
	\end{align}
	respectively.
	
	In Fig. \ref{fig: Np} (a), we measure the auto-correlation and cross-correlation of the Zadoff-chu sequence with respect to the pilot length. 
	The auto-correlation value is averaged out for the cases of $\tau_1\ne\tau_2$. The auto-correlation and cross-correlation are normalized with $N_p$ so that the auto-correlation for the case of $\tau_1=\tau_2$ is $1$. We observe that the auto-correlation and cross-correlation decrease as the pilot length increases. Hence, we implement the Zadoff-chu sequence with sufficiently long pilot lengths, which approximates to the idealistic properties given in (\ref{eq: 47}) and (\ref{eq: 48}). In Fig. \ref{fig: Np} (b), the DER and NMSE are measured with respect to the pilot length to verify the performance of the proposed channel estimation technique. The DER decreases as the pilot length increases and eventually saturates to zero. This is because the DER performance strongly depends on the idealistic properties in \eqref{eq: 47} and (\ref{eq: 48}). The NMSE of the channel gain also decreases as the pilot length increases since the matrix $\hat{\boldsymbol{\Psi}}$ for the LS estimation is highly related to the DER performance.

	In Fig. \ref{fig:Channel est}, we investigate the DER and NMSE of the proposed channel estimation technique with respect to the transmit power of the mUEs. We set the transmit power of every mUE to be equal and the pilot length as 100. We also investigate three different cases by changing the number of mUEs as 5, 10, and 15. In Fig. \ref{fig:Channel est} (a), the DER decreases as the transmit power increases for all three cases due to the increasing signal-to-interference-plus-noise ratio (SINR). However, in the high transmit power regime, the DER tends to saturate since the SINR saturates as the transmit power increases.  
	As the number of mUEs increases, the DER increases because the interference term in \eqref{eq: 45} increases due to the degradation of the cross-correlation property. For the NMSE in Fig. \ref{fig:Channel est} (b), all three cases decrease as the transmit power increases. In the high power regime, the cases with the number of mUEs of $10$ and $15$ saturate because the DER directly affects the channel gain estimation through the estimated delays $\left\{ \hat{\tau}_1, \cdots,\hat{\tau}_K \right\}$.

	\begin{figure}%	
		\subfloat[The DER performance with respect to the power of the mUEs.]{{\includegraphics[width=0.5\textwidth ]{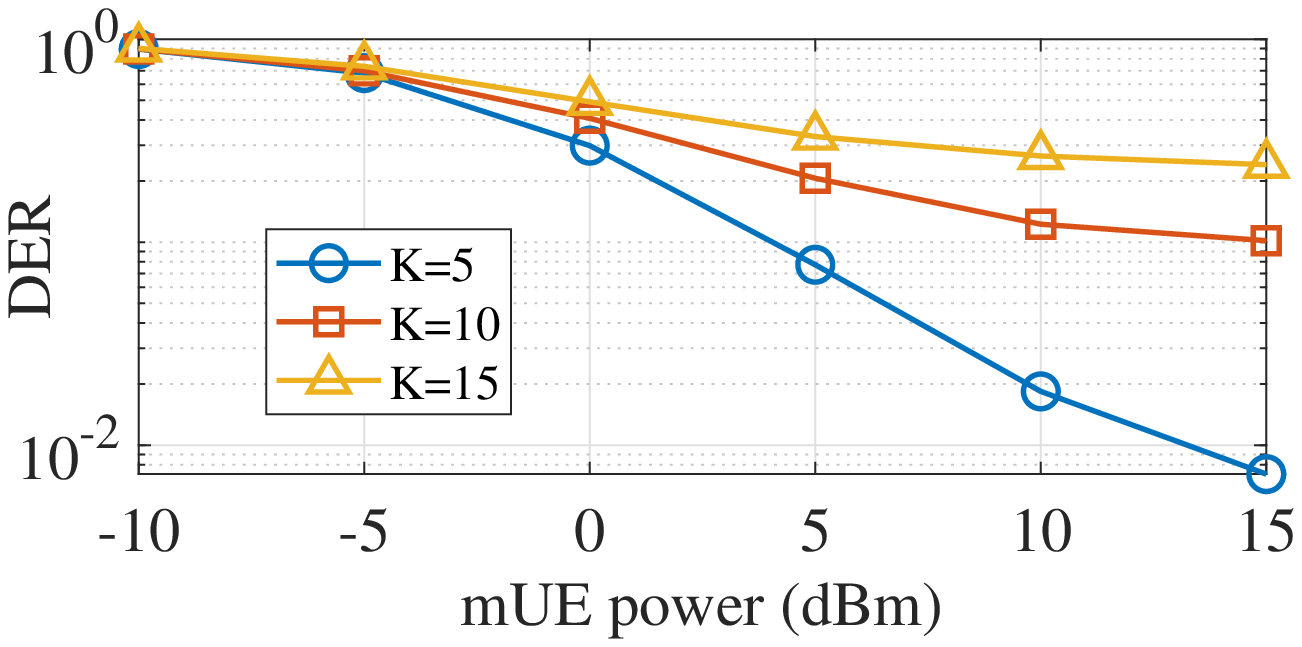}}}% 		
		\hfil
		\subfloat[The NMSE performance with respect to the power of the mUEs.]{{\includegraphics[width=0.5\textwidth ]{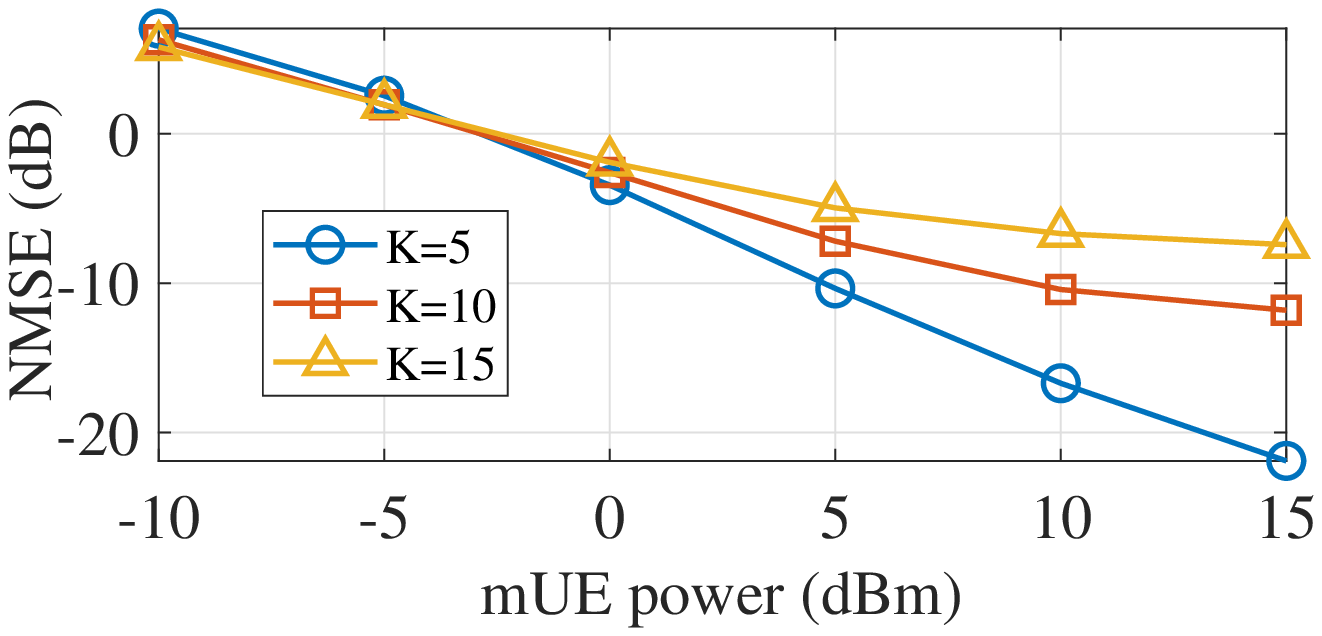}}}% 					
		\caption{Performances of the proposed channel estimation technique with respect to the power of the mUEs, where $N_t=16$.}% 		
		\label{fig:Channel est}% 		
	\end{figure}
	
	The cases of eCRS with one SIC layer are denoted as IeCRS, DeCRS, and LOW, where IeCRS and DeCRS adopt the convex optimization approach in IeCRS and DeCRS cases, respectively, and LOW adopts the low complexity approach for IeCRS. 
	To compare the results of our proposed framework, we propose three types of benchmarks, namely, identical cooperative NOMA (IC-NOMA), distinct cooperative NOMA (DC-NOMA), and single tap (ST). IC-NOMA and DC-NOMA are IeCRS and DeCRS without using common messages, similar to NOMA, respectively. 
	ST is the ideal case where the signals from the mUEs arrive simultaneously while using IeCRS. 
	The precoders for the first phase of IeCRS, LOW, IC-NOMA, and ST are initialized by using maximum ratio transmission (MRT) combined with singular value decomposition (SVD) as in \cite{RSMA2}. The precoders for the first phase of DeCRS and DC-NOMA are initialized by using MRT for both the private and common streams. The second phase for all cases is initialized by the mUEs using all their power.
	Although these benchmarks are valid, they are all special cases of our proposed framework.
	To the best of our knowledge, we could not find any adequate benchmark that could adapt to our scenario of interest. 
	Thus, we verify our framework with the benchmarks stated above.
	
	\begin{figure}[t] 
		\centering
		\includegraphics[width=1 \columnwidth]{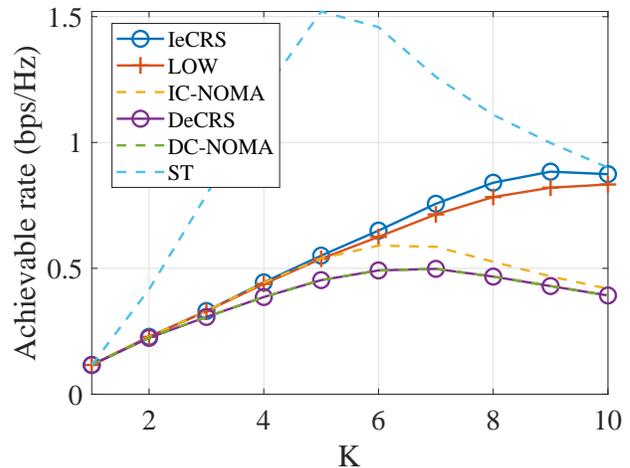}
		%   % where an .eps filename suffix will be assumed under latex,
		%   % and a .pdf suffix will be assumed for pdflatex
		\caption{The achievable rate with respect to the number of mUEs $K$, where $N_t = 32$.} 
		\label{fig: AboutK_Ben}
	\end{figure}
	
	In Fig. \ref{fig: AboutK_Ben}, we plot the achievable rate with respect to the number of mUEs. 
	Unlike conventional transmission techniques, where the rate decreases as the number of mUEs increases, we observe that there exists a certain performance peak for our proposed framework.
	This peak is due to the two-phase nature of eCRS. For the first phase, the transmission is similar to the conventional MU-MISO downlink channel, where the achievable rate decreases as the number of mUEs increases. 
	For the second phase, as the number of mUEs increases, the achievable rate of the dUE naturally increases. 
	In result, the performance increases for a small number of mUEs, where the performance bottleneck is from the second phase, and the performance decreases for a large number of mUEs due to the bottleneck of the first phase. 
	We observe that IeCRS and LOW have similar performances for small and large numbers of mUEs. 
	This is expected, as our derivation was from a low SNR assumption, thus, the performance gap for a small number of mUEs is expected to be tight. 
	The performance gap steadily increases as the number of mUEs increases and becomes tight again since the first phase becomes the bottleneck. 
	DeCRS seems to have lower performance compared to IeCRS and LOW. 
	This is due to two factors. First, the achievable rate that we consider is already a lower bound. 
	Similar to IeCRS, different signals from the mUEs arrive in different instances in general. 
	Since we neglected this factor, the performance is degraded. 
	Also, DeCRS adopts more streams for both the first and second phases.
	Since the mUEs and dUE are all equipped with a single antenna, the interference is crucial. We expect that with multiple antenna mUEs and dUEs, the performance of DeCRS will increase drastically, which is left for future work. 
	Finding the balance of the interference and noise through the common message of the mUEs, we observe that IeCRS outperforms IC-NOMA. 
	However, as expected, DeCRS has the same performance as DC-NOMA since the common messages do not control the interference between the mUEs.
	We observe that due to the signals arriving in different taps, the performance of IeCRS is upper bounded with ST, showing that there is an inevitable performance loss to consider fast-sampling communication systems.
	Note that, the signals arriving in different taps is beneficial for DeCRS, since the interference will decrease. However, this will be detrimental for IeCRS, since this would decrease the strength of the overall signal.
	
%	\begin{figure}[t] 
%		\centering
%		\includegraphics[width=1 \columnwidth]{prop_bench_power.eps}
%		\caption{The achievable rate with respect to the power of mUEs, where the downlink channels are heterogeneous, $K = 10$ and $N_t=16$.} 
%		\label{fig: AboutP_Ben}
%	\end{figure}
	
%	In Fig. \ref{fig: AboutP_Ben}, we investigate how the heterogeneous downlink channels affect the achievable rate with respect to the transmit power of the mUEs.
%	Among ten mUEs, we assume that the downlink channel for eight mUEs experience additional power loss of -20 dB.
%	The achievable rate for our framework increases until it saturates in the high power regime. 
%	This is because the performance of the two-phase C-RSMA is bounded by the first phase in the high power regime. 
%	While the achievable rate of the second phase increases as the power of the mUEs increases, the achievable rate of the first phase is not affected. In the scenario of heterogeneous downlink channels, we observe that NDCDT outperforms NICDT.
%	This is because the achievable rate of the dUE for NICDT is bounded by the weak channel mUEs in the first phase. 
%	On the other hand, ICDT outperforms DCDT, since ICDT can fully exploit the benefit of common data compared to DCDT. Hence, the gap between NICDT is significant compared to the gap between DCDT and NDCDT.  
%	This shows the superior performance of the concept of common data.
	
	\begin{figure}[t] 
		\centering
		\includegraphics[width=1 \columnwidth]{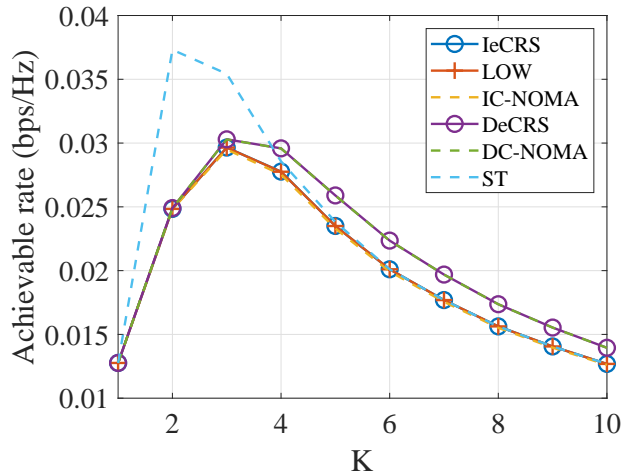}
		\caption{The achievable rate with respect to the number of mUEs $K$, where $N_t = 16$, $P_{\mathrm{AP}} = -10$ dBm, and $P_k = -10$ dBm.} 
		\label{fig: AboutK_SmallP}
	\end{figure}

	In Fig. \ref{fig: AboutK_SmallP}, we plot the achievable rate with respect to the number of mUEs, where the AP and mUEs have low transmit power. Similar to Fig. \ref{fig: AboutK_Ben}, all cases have performance peaks, where the peaks are shifted to the left due to the low transmit power of the AP.
	Another result to note is that DeCRS outperforms IeCRS.
	For DeCRS, the AP can concentrate the power to specific mUEs by selecting a few mUEs for cooperation.
	We also observe that there is no significant performance improvement of IeCRS compared to IC-NOMA, since the benefit of the common message that finds the balance between the interference and noise is limited due to the dominant AWGN.
%	IC-NOMA even shows a higher achievable rate than IeCRS when there are many mUEs, which is unexpected since IeCRS is a superset of IC-NOMA.
%	However, this is possible because the WMMSE approach in IeCRS can converge to a local optimum.

	\begin{figure}[t] 
		\centering
		\includegraphics[width=1 \columnwidth]{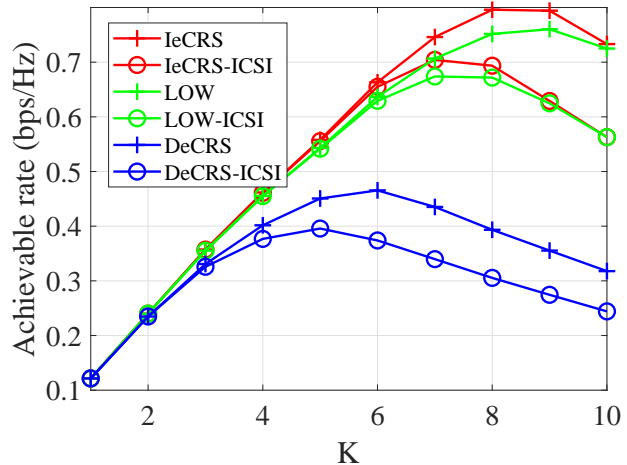}
		%   % where an .eps filename suffix will be assumed under latex,
		%   % and a .pdf suffix will be assumed for pdflatex
		\caption{The achievable rate with respect to the number of mUEs $K$, where $N_t = 16$ and $N_p = 500$.}
		\label{fig: AboutK_ICSI}
	\end{figure}
	
	In Fig. \ref{fig: AboutK_ICSI}, we compare the performances of eCRS for perfect CSI and imperfect CSI cases with respect to the number of mUEs.
	%We also simulate cases where the decoding order of the common and private data are reversed, denoted as REV.
	Similar to Fig. \ref{fig: AboutK_Ben}, the achievable rate seems to have a peak for eCRS, where the performance bottleneck shifts from the second phase to the first phase as the number of mUEs increases.
	%We observe that for ICDT, the decoding order affects the performance significantly, while for DCDT, the performance is not affected.
	%ICDT controls the common data for all the mUEs, while each common data in DCDT controls only a single mUE.
	%Thus, the degree-of-freedom (DoF) advantage from the decoding order appears harshly in ICDT-REV than DCDT-REV.
	For all cases, the imperfect CSI cases show performance gaps between the perfect CSI cases as the number of mUEs increases. This is related to the performance of channel estimation, where the estimation error increases as the number of mUEs increases.

	\section{Conclusion} \label{sec:concl}
	In this paper, we proposed a novel message transmission framework to increase the coverage for a THz MU-MISO downlink system through cooperative communication using RSMA. 
	For message transmissions, we proposed eCRS, and explored two specific cases, which are IeCRS and DeCRS.
	Based on the novel THz cooperative channel model, we derived local optimal solutions of IeCRS and DeCRS through convex optimization techniques as well as a closed form solution for IeCRS in the low SNR regime.
	Finally, to successfully use our framework in practice, we proposed a channel estimation technique to detect the channel gains and time delays of the THz cooperative channel model. 
	Through simulation results, we confirmed that our proposed message transmission framework has considerable performance, and that our estimation technique successfully captures the full capabilities of the THz cooperative channel model.
	
	%\appendix
	%%%%%%%%%%%%%%%%
	%% Background %%
	%%

	\bibliographystyle{IEEEtran}
	% argument is your BibTeX string definitions and bibliography database(s)
	\bibliography{HyebibTHz}
	
\end{document}